\documentclass[11pt]{article}
\usepackage[margin=0.75in]{geometry}
\usepackage{color}
\usepackage{graphicx}
\usepackage{amsmath}
\usepackage{amssymb}
\usepackage{amsthm}
\usepackage{subfigure}
\usepackage{rotating}
\usepackage{verbatim}
\usepackage{hyperref}
\usepackage{bm}
\usepackage[authoryear]{natbib}
\usepackage{bbm}
\usepackage{xr}
\usepackage{array}
\usepackage{tikz}

\usepackage{enumerate,algorithm2e}

\usepackage{float}

\newcommand{\one}[1]{{\mathbbm{1}}_{{#1}}}

\newcommand{\PP}[1]{\mathbb{P}\left\{{#1}\right\}} 

\renewcommand{\O}[1]{\mathcal{O}\left({#1}\right)}
\def\R{\mathbb{R}}

\def\independenT#1#2{\mathrel{\rlap{$#1#2$}\mkern2mu{#1#2}}}
\newcommand\independent{\protect\mathpalette{\protect\independenT}{\perp}}
\newcommand{\iidsim}{\stackrel{\mathrm{iid}}{\sim}}

\newcommand{\ignore}[1]{}

\newcommand{\E}[1]{\mbox{$\mathbb{E}\left(#1\right)$}}

\newcommand{\iid}{\stackrel{\emph{i.i.d.}}{\sim}}
\newcommand{\bs}{\boldsymbol}

\DeclareMathOperator*{\sign}{sign}

\newcommand{\bX}{\bs{X}}
\newcommand{\bG}{\bs{G}}
\newcommand{\bS}{\bs{\Sigma}}
\newcommand{\bXk}{\tilde{\bs{X}}}
\newcommand{\by}{\bs{y}}

\newcommand{\Zk}{\tilde{Z}}

\graphicspath{{./fig/}}

\numberwithin{equation}{section}

\definecolor{darkred}{rgb}{0.6, 0.0, 0.0}
\usepackage[normalem]{ulem}

\newcommand{\ljcc}[1]{{\color{red} [LJ: {#1}]}}

\definecolor{ejc}{RGB}{0,0,200}
\newcommand{\ejc}[1]{\textcolor{ejc}{[EJC: #1]}}

\newcommand{\rev}[1]{{\color{black}#1}}
\newcommand{\revv}[1]{{\color{black}#1}}
\usepackage{authblk}
\title{Panning for Gold:\\
Model-X Knockoffs for High-dimensional Controlled Variable
Selection}
\author[1]{Emmanuel Cand\`es\thanks{Author names are sorted
    alphabetically.}}
\author[2]{Yingying Fan}
\author[1]{Lucas Janson}
\author[2]{Jinchi Lv}
\date{}
\affil[1]{Department of Statistics, Stanford University}
\affil[2]{Data Sciences and Operations Department, Marshall School of
  Business, USC}

\begin{document}
\maketitle

\begin{abstract}
  Many contemporary large-scale applications involve building
  interpretable models linking a large set of potential covariates to
  a response in a nonlinear fashion, such as when the response is
  binary. Although this modeling problem has been extensively studied,
  it remains unclear how to effectively control the fraction of false
  discoveries even in high-dimensional logistic regression, not to
  mention general high-dimensional nonlinear models. To address such a
  practical problem, we propose a new framework of
  \textit{model-\revv{X}} knockoffs, which reads from a different
  perspective the knockoff procedure \citep{RB-EC:2015} originally
  designed for controlling the false discovery rate in linear
  models. \revv{Whereas the knockoffs procedure is constrained to
    homoscedastic linear models with $n\ge p$, the key innovation here
    is that model-\revv{X} knockoffs provide valid inference from
    finite samples in settings in which the conditional distribution
    of the response is arbitrary and completely unknown. Furthermore,
    this holds no matter the number of covariates. Correct inference
    in such a broad setting is achieved by constructing knockoff
    variables probabilistically instead of
    geometrically.} 
  \revv{To do this, our} approach requires \revv{the covariates be
      random (independent and identically distributed rows) with a
      distribution} that is known, although we \revv{provide preliminary experimental evidence that our procedure is} robust to unknown/estimated distributions. To our knowledge, no
    other procedure solves the \textit{controlled} variable selection
    problem in such generality, but in the restricted settings where
    competitors exist, we demonstrate the superior power of knockoffs
    through simulations. Finally, we apply our procedure to data from
    a case-control study of Crohn's disease in the United Kingdom,
    making twice as many discoveries as the original analysis of the
    same data.

\smallskip
\noindent \textbf{Keywords.} False discovery rate (FDR), knockoff
filter, testing for conditional independence in nonlinear models,
generalized linear models, logistic regression, Markov blanket,
genome-wide association study (GWAS)
\end{abstract}

\section{Introduction}
\label{sec:introduction}
\subsection{Panning for gold}
\label{sec:overview}
Certain diseases have a genetic basis, and an important biological
problem is to find which genetic features (e.g., gene expressions or
single nucleotide polymorphisms) are important for determining a given
disease. In health care, researchers often want to know which
electronic medical record entries determine future medical
costs. Political scientists study which demographic or socioeconomic
variables determine political opinions. Economists are similarly
interested in which demographic/socioeconomic variables affect future
income. Those in the technology industry seek out specific software
characteristics they can change to increase user engagement. In the
current data-driven science and engineering era, a list of such
problems would go on and on. The common theme in all these instances
is that we have a deluge of explanatory variables, often many more
than the number of observations, knowing full well that the outcome we
wish to understand better only actually depends on a small fraction of
them. Therefore, a primary goal in modern ``big data analysis'' is to
identify those important predictors in a sea of noise
variables. Having said this, a reasonable question is why do we have
so many covariates in the first place? The answer is twofold: first,
because we can. To be sure, it may be fairly easy to measure
thousands if not millions of attributes at the same time. For
instance, it has become relatively inexpensive to genotype an
individual, collecting hundreds of thousands of genetic variations at
once. Second, even though we may believe that a trait
or phenotype depends on a comparably small set of genetic
variations, we have a priori no idea about which ones are
relevant and therefore must include them all in our search for
those nuggets of gold, so to speak. To further complicate matters, a
common challenge in these big data problems, and a central focus of
this paper, is that we often have little to no knowledge of how the
outcome even depends on the few truly important variables.

To cast the ubiquitous \emph{(\revv{variable}) selection} problem 
in statistical terms, call $Y$ the random variable representing the
outcome whose determining factors we are interested in, and
$X_1,\dots,X_p$ the set of $p$ potential determining factors or explanatory variables. The
object of study is the \emph{conditional} distribution of the outcome
$Y$ given the covariates $X = (X_1,\dots,X_p)$, and we shall denote
this conditional distribution function by $F_{Y|X}$. Ideally we would
like to estimate $F_{Y|X}$, but in general, this is effectively impossible
from a finite sample. For instance, even knowing that the conditional
density depends upon 20 {\em known} covariates makes the problem
impossible unless either the sample size $n$ is astronomically large,
and/or we are willing to impose a very restrictive model.
However, in most problems $F_{Y|X}$ may realistically be assumed to
depend on a small fraction of the $p$ covariates; that is, the
function $F_{Y|X}(y|x_1, \ldots, x_p)$ only depends upon a small number
of coordinates $x_i$ (or is well approximated by such a
lower-dimensional function).
Although this assumption does not magically make the estimation of
$F_{Y|X}$ easy, it does suggest consideration of the simpler problem:
{\em which of the many variables does $Y$ depend upon?} Often, finding
a few of the important covariates---in other words, teasing out the
relevant factors from those which are not---is already scientifically
extremely useful and can be considered a first step in understanding
the dependence between an outcome and some interesting variables; we
regard this as a crucial problem in modern data science.

\subsection{A peek at our contribution}
\label{sec:problem}
This paper addresses the selection problem by considering a very
general conditional model, where the response $Y$ can depend in an
arbitrary fashion on the covariates $X_1,\dots,X_p$. The only
restriction we place on the model is that the observations
$(X_{i1},\dots,X_{ip}, Y_i)$ are independently and identically
distributed (i.i.d.), which is often realistic in high-dimensional
applications such as genetics, where subjects may be drawn randomly
from some large population, or client behavioral modeling, where
experiments on a service or user interface go out to a random subset
of users. Therefore, the model is simply
\begin{equation}
\label{eq:glm}
(X_{i1},\dots,X_{ip}, Y_i) \iid F_{XY},\;\;\; i = 1,\dots,n
\end{equation}
for some arbitrary $(p+1)$-dimensional joint distribution $F_{XY}$. We
will assume \emph{no knowledge} of the conditional distribution of
$Y\,|\,X_1,\dots,X_p$, but we do assume the joint distribution of the
covariates is known, and we will denote it by $F_X$. As a
  concrete example, consider a case-control experiment to determine
  the genetic factors which contribute to a rare disease, with
  diseased subjects oversampled to 50\% to increase power. Then the
  joint distribution of features and disease status obeys the model
  \eqref{eq:glm}, where $F_X$ is a 50\%--50\% mixture of the genetic
  distributions of diseased and healthy subjects, and $Y$ is the
  subjects' binary disease state.

In Section \ref{sec:model} below we shall discuss the merits of this
model but we would like to immediately remark on an important benefit:
namely, one can pose a meaningful problem. To do this, observe that
when we say that the conditional distribution of $Y$ actually depends
upon a (small) subset $\mathcal{S} \subset \{1, \ldots, p\}$ of the
variables $X_1, \ldots, X_p$, which we would like to identify, we mean
that we would like to find the ``smallest'' subset $\mathcal{S}$ such
that conditionally on $\{X_j\}_{j \in \mathcal{S}}$, $Y$ is
independent of all other variables. Another way to say this is that
the other variables do not provide additional information about $Y$. A
minimal set $\mathcal{S}$ with this property is usually called a
Markov blanket or Markov boundary for $Y$ in the literature on
graphical models \citep[Section 3.2.1]{JP:1988}. Under very mild
conditions about the joint distribution $F_{XY}$, the Markov blanket
is well defined and unique (see Section \ref{sec:problem_statement}
for details) so that we have a cleanly stated selection
problem. Note also that the Markov Blanket can be defined purely
  in terms of $F_{Y|X}$ without any reference to $F_X$, so that in our
  case-control example the problem is defined in exactly the same way
  as if $F_X$ were the true population genetic distribution instead of
  the oversampled mixture of diseased and healthy genetic
  distributions.

In most problems of interest, even with the knowledge of $F_X$, it is
beyond hope to recover the blanket $\mathcal{S}$ with no error. Hence,
we are naturally interested in procedures that control a Type I error;
that is, we would like to find as many variables as possible while at
the same time not having too many false positives.  In this paper, we
focus on controlling the false discovery rate (FDR)
\citep{YB-YH:1995}, which we can define here as follows: letting
$\hat{\mathcal{S}}$ be the outcome of a selection procedure operating
on the sampled data (we put a hat because $\hat{\mathcal{S}}$ is
random), the FDR is
\begin{equation}
\label{eq:fdr0}
\text{FDR} :=  \mathbb{E}[\text{FDP} ], \qquad \text{FDP} = \frac{\# \{j: j
  \in \hat{\mathcal{S}}\setminus \mathcal{S}\}}{\# \{j: j
  \in \hat{\mathcal{S}}\}}
\end{equation}
with the convention 0/0 = 0. Procedures that control the FDR are
interpretable, as they roughly bound what fraction of discoveries are
false ones, and they can be quite powerful as well.

One achievement of this paper is to show that we can design rather
powerful procedures that rigorously control the FDR \eqref{eq:fdr0} in
finite samples. This holds no matter the unknown relationship between
the explanatory variables $X$ and the outcome $Y$. We achieve this by
rethinking the conceptual framework of \citet{RB-EC:2015}, who
originally introduced the knockoff procedure (throughout this paper,
we will sometimes use ``knockoffs'' as shorthand for the knockoff
framework or procedure). Their salient idea was to construct a set of
so-called ``knockoff'' variables which were not (conditionally on the
original variables) associated with the response, but whose structure
mirrored that of the original covariates. These knockoff variables
could then be used as controls for the real covariates, so that only
real covariates which appeared to be considerably more associated with
the response than their knockoff counterparts were selected.  Their
main result was achieving exact finite-sample FDR control in the
homoscedastic Gaussian linear model when $n\ge 2p$ (along with a
nearly-exact extension to when $p\le n<2p$). By reading the knockoffs
framework from a new perspective, the present paper places \emph{no
  restriction} on $p$ relative to $n$, in sharp constrast to the
original knockoffs work which required the low-dimensional setting of
$n\ge p$. The conceptual difference is that the original knockoff
procedure treats the $X_{ij}$ as fixed and relies on specific
stochastic properties of the linear model, precluding consideration of
$p>n$ or nonlinear models. By treating the $X_{ij}$ as \emph{random}
and relying on that stochasticity instead, the model-\revv{X}
perspective allows treatment of the high-dimensional setting which is
increasingly the norm in modern applications.  We refer to the new
approach as \emph{model-\revv{X}} (M\revv{X}) knockoffs\revv{, and by
  contrast we refer to the original knockoffs approach of
  \cite{RB-EC:2015} as \emph{fixed-X} (FX) knockoffs}. In a nutshell:
\begin{itemize}
\item We propose a new knockoff construction amenable to the \revv{random
  covariate setting} \eqref{eq:glm}.
\item As in \citet{RB-EC:2015} and further reviewed in
  Section~\ref{sec:methodology}, we shall use the knockoff variables
  as controls in such a way that we can tease apart important
  variables from noise while controlling the FDR. \revv{\em We place no
    restriction on the dimensionality of the data or the conditional
    distribution of $Y\,|\,X_1,\dots,X_p$.}
\item We apply the new procedure to real data from a case-control
  study of Crohn's disease in the United Kingdom, please see Section
  \ref{sec:realdata}. There, we show that the new knockoff method
  makes twice as many discoveries as the original analysis of the
  same data.
\end{itemize}

Before turning to the presentation of our method and results, we pause
to discuss the merits and limitations of our model, the relationships
between this work and others on selective inference, and the larger
problem of high-dimensional statistical testing. 

\subsection{Relationship with the classical setup for inference}
\label{sec:model}

It may seem to the statistician that our model appears rather
different from what she is used to. Our framework is, however,
not as exotic as it looks.

\paragraph{Classical setup}
The usual setup for inference {in conditional models} is to assume a
strong parametric model for the response conditional on the
covariates, such as a homoscedastic linear model, but to assume as
little as possible about, or even condition on, the covariates. We do
the exact opposite by assuming we know \emph{everything} about the
covariate distribution but \emph{nothing} about the conditional
distribution $Y|X_1,\dots,X_p$.  Hence, we merely shift the
\emph{burden of knowledge.}  \revv{Our philosophy is, therefore, 
  to model $X$, not $Y$, whereas
  classically, $Y$ (given $X$) is modeled and $X$ is not.} In
practice, the parametric model in the classical approach is just an
approximation, and does not need to hold exactly to produce useful
inference. Analogously, we do not need to know the covariate
distribution exactly for our method to be useful, as we will
demonstrate in Sections~\ref{sec:robust} and \ref{sec:realdata}.

\paragraph{When are our assumptions useful?}
We do not claim our assumptions will always be appropriate, but there
are important cases when it is reasonable to think we know much more
about the covariate distribution than about the conditional
distribution of the response, including:
\begin{itemize}
	\item When we in fact know exactly the covariate distribution because
	we control it, such as in gene knockout experiments
	\citep{LC-ea:2013,JP-ea:2016}, genetic crossing experiments \citep{JH-CW:1931},
	or sensitivity analysis of numerical models \citep{AS-ea:2008} (for
	example climate models). In some cases we may also essentially know
	the covariate distribution even when we do not control it, such as
	in admixture mapping \citep{HT-ea:2006}.
	\item When we have a large amount of unsupervised data (covariate data without
	corresponding responses/labels) in addition to the $n$ labeled
	observations. This is not uncommon in genetic or economic studies,
	where many other studies will exist that have collected the same
	covariate information but different response variables.
	\item When we simply have considerably more prior information about
	the covariates than about the response. Indeed, the point of many
	conditional modeling problems is to relate a poorly-understood
	response variable to a set of well-understood covariates. For
	instance, in genetic case-control studies, scientists seek to
	understand the \rev{genetic basis} of an extremely biologically-complex disease
	using many comparatively simple single nucleotide polymorphisms
	(SNPs) as covariates.
\end{itemize}

\paragraph{Payoff} There are substantial payoffs to our
framework. Perhaps the main advantage is the ability to use the
  knockoff framework in high dimensions, a setting that was impossible
  using the original approach. Even in high dimensions, previous inference results
rely not only on a parametric model that is often linear and
homoscedastic, but also on the sparsity or ultra-sparsity of the
parameters of that model in order to achieve some asymptotic
guarantee. In contrast, our framework can accommodate \emph{any} model
for both the response and the covariates, and our guarantees are exact
in finite samples (non-asymptotic). In particular, our setup
encompasses any regression, classification, or survival model, including any
generalized linear model (GLM), and allows for arbitrary
nonlinearities and heteroscedasticity, such as are found in many
machine learning applications.

\subsection{Relationship with work on inference after selection}
\label{sec:related}

{There is a line of work on inference after selection, or
	post-selection inference, for high-dimensional regression, the goal
	of which is to first perform selection to make the problem
	low-dimensional, and then produce p-values that are valid
	\emph{conditional} on the selection step
	\citep{RB-ea:2013,RL-ea:2014,JL-ea:2016}. These works differ from
	ours in a number of ways so that we largely see them as
	complementary activities.
\begin{itemize}
\item First, our focus is on selecting the right variables, whereas
  the goal of this line of work is to adjust inference after some
  selection has taken place. In more detail, these works presuppose a
  selection procedure has been chosen (for reasons that may have
  nothing to do with controlling Type I error) and then compute
  p-values for the selected variables, taking into account the
  selection step. In contrast, {M\revv{X}} knockoffs is by itself a selection
  procedure that controls Type I error.
		
\item Second, inference after selection relies heavily on parametric
  assumptions about the conditional distribution, namely that
  \[Y|X_1,\dots,X_p \sim \mathcal{N}(\mu(X_1,\dots,X_p),\sigma^2),\]
  making it unclear how to extend it to the more general setting of
  the present paper.
		
\item The third difference stems from their objects of inference. In
  the selection step, a subset of size $m\le n$ of the original $p$
  covariates is selected, say $X_{j_1},\dots,X_{j_m}$, and the objects
  of inference are the coefficients of the $X_{j_k}$'s in the
  projection of $\mu$ onto the linear subspace spanned by the
  $n\times m$ matrix of observed values of these $X_{j_k}$'s. That is,
  the $k$th null hypothesis is that the aforementioned coefficient on
  $X_{j_k}$ is zero---note that whether or not inference on the $j$th
  variable is produced at all, and if it is, the object of that
  inference, both depend on the initial selection step. In contrast,
  if {M\revv{X}} knockoffs were restricted to the homoscedastic Gaussian
  model above, the $j$th null hypothesis would be that $\mu$ does not
  depend on $X_j$, and there would be $p$ such null hypotheses, one
  for each of the original variables.
\end{itemize}

	\subsection{Obstacles to obtaining p-values}
	\label{sec:obs}
	
	Our procedure does not follow the canonical approach to FDR control
	and multiple testing in general.  The canonical approach is to plug
	p-values into the BHq procedure, which controls the FDR under p-value
	independence and certain forms of dependence
	\citep{YB-YH:1995,YB-DY:2001}.
	Although these works have seeded a wealth of methodological
        innovations over the past two decades \citep{YB:2010}, all
        these procedures act on a set of valid p-values (or equivalent
        statistics), which they assume can be computed.\footnote{\rev{
            \citet{YB-YG:2009} and \citet{MB-ea:2015} transform the
            p-value cutoffs of common FDR-controlling procedures into
            penalized regression analogues to avoid p-values
            altogether. They only} provably control the FDR in
          homoscedastic linear regression when the design matrix has
          orthogonal columns (necessitating, importantly, that
          $n\ge p$) \rev{but the latter work} empirically \rev{retains
            control more generally} whenever the signal obeys sparsity
          constraints. \rev{In a very different setting with spatial
            hypotheses, \citet{JL-ea:2016b} use approximations from
            Gaussian random field theory to heuristically control the FDR.}}
        The requirement of having valid p-values is quite constraining
        for general {conditional modeling} problems.
	
	\subsubsection{Regression p-value approximations}
        \label{sec:regrp}
	In low-dimensional ($n\ge p$) homoscedastic Gaussian linear
	regression, p-values can be computed exactly even if the error
	variance is unknown, although the p-values will not in general have
	any simple dependence properties like independence or positive
	regression dependency on a subset (PRDS). Already
	for just the slightly broader class of low-dimensional GLMs, one must
	resort to asymptotic p-values derived from maximum-likelihood theory,
	which we will show in \rev{Section~\ref{sec:logreg}} can be far from valid
	in practice. 
        In high-dimensional ($n<p$) GLMs, it is not clear how to get
        p-values at all. Although some work (see for example
        \citet{SV-ea:2014}) exists on computing asymptotic p-values
        under strong sparsity assumptions \rev{(usually the number of
          important variables must be $o(\sqrt{n}/\log(p))$)}, like
        their low-dimensional maximum-likelihood counterparts, these
        methods suffer from highly non-uniform null p-values in many
        finite-sample problems \rev{(see, for example, simulations in
          \cite{RD-ea:2015})}. For binary covariates, the causal
        inference literature uses matching and propensity scores for
        approximately valid inference, but extending these methods to
        high dimensions is still a topic of current research,
        requiring similar assumptions and asymptotic approximations to
        the aforementioned high-dimensional GLM literature
        \citep{SA-GI-SW:2016}. Moving beyond generalized linear models
        to the nonparametric setting, there exist measures of feature
        importance, but no p-values.\footnote{In their online
          description of random forests
          (\url{http://www.math.usu.edu/~adele/forests/}), Leo Breiman
          and Adele Cutler propose a way to obtain a ``z-score'' for
          each variable, but without any theoretical distributional
          justification, and \cite{CS-AZ:2008} find ``that the
          suggested test is not appropriate for statements of
          significance.''}
	
	\subsubsection{Marginal testing}
        \label{sec:mrgnl}
	Faced with the inability to compute p-values for hypothesis tests of
	conditional independence, one solution is to use \emph{marginal}
	p-values, i.e., p-values for testing \emph{un}conditional (or
	marginal) independence between $Y$ and $X_j$. This simplifies the
	problem considerably, and many options exist for obtaining valid
	p-values for such a test. However, marginal p-values are in general
	{\em invalid} for testing conditional independence, and replacing
	tests of conditional independence with tests of unconditional
	independence is often undesirable. Indeed when $p \ll n$, so that
	classical (e.g., maximum-likelihood) inference techniques for
	regression give valid p-values for parametric tests of conditional
	independence, it would be very unusual to resort to marginal testing
	to select important covariates, and we cannot think of a textbook
	that takes this route. Furthermore, the class of conditional test
	statistics is far richer than that of marginal ones, and includes the
	most powerful statistical inference and prediction methodology
	available. For example, in compressed sensing, the signal recovery
	guarantees for state-of-the-art $\ell_1$-based (joint) algorithms are
	stronger than any guarantees possible with marginal methods. To
	constrain oneself to marginal testing is to completely ignore the vast
	modern literature on sparse regression that, while lacking
	finite-sample Type I error control, has had tremendous success
	establishing other useful inferential guarantees such as model
	selection consistency under high-dimensional asymptotics in both
	parametric (e.g., Lasso \citep{PZ-BY:2006,CandesPlan}) and
	nonparametric (e.g., random forests \citep{SW-SA:2016})
	settings. Realizing this, the statistical genetics community has
	worked on a number of multivariate approaches to improve power in
	genome-wide association studies using both penalized
	\citep{TW-ea:2009,QH-DL:2011} and Bayesian regression
	\citep{YG-MS:2011,JL-ea:2011}, but both approaches still suffer from a
	lack of Type I error control (without making strong assumptions on
	parameter priors). We will see that the {M\revv{X}} knockoff procedure
	is able to leverage the power of any of these techniques while adding
	rigorous finite-sample Type I error control \revv{when the covariate distribution is known}.

Some specific drawbacks of using marginal p-values are:
\begin{description}
\item[1. Power loss] Even when the covariates are independent, so that
  the hypotheses of conditional and unconditional independence
  coincide, p-values resulting from marginal testing procedures may be
  less powerful than those from conditional testing procedures. This
  phenomenon has been reported previously, for example in statistical
  genetics by \citet{CH-ea:2008} and many others. Intuitively, this is
  because a joint model in $X_1,\dots,X_p$ for $Y$ will have less
  residual variance than a marginal model in just $X_j$.
  There are exceptions, for instance if there is only one important
  variable, then its marginal model is the correct joint model and a
  conditional test will be less powerful due to the uncertainty in how
  the other variables are included in the joint model. But in general,
  marginal testing becomes increasingly underpowered relative to
  conditional testing as the \emph{absolute} number of important
  covariates increases \citep{FF-ea:2012}, suggesting particular
  advantage for conditional testing in modern applications with
  complex high-dimensional models.
		
  There are also cases in which important variables are in fact fully
  marginally independent of the response. As a toy example, if
  $X_1,X_2\iidsim\text{Bernoulli}(0.5)$ and $Y = \one{\{X_1+X_2=1\}}$,
  then $Y$ is marginally independent of each of $X_1$ and $X_2$, even
  though together they determine $Y$ perfectly. $X_1$ and $X_2$ are
  certainly \emph{conditionally} dependent on $Y$, however, so a
  conditional test can have power to discover them.
		
\item[2. Interpretability] When the covariates are not independent,
  marginal and conditional independence do not correspond, and we end
  up asking the wrong question. For example, in a model with a few
  important covariates which \rev{determine} the outcome, and many
  unimportant covariates which have no influence on the outcome but
  are correlated with the \rev{important} covariates, marginal testing
  will treat such unimportant covariates as important. Thus, because
  marginal testing is testing the wrong hypotheses, there will be many
  ``discoveries'' which have no influence on the outcome.
		
  The argument is often made, especially in genetics, that although
  discovering a\rev{n unimportant} covariate just because it was
  correlated with a\rev{n important} one is technically incorrect, it
  can still be useful as it suggests that there is a\rev{n important}
  covariate correlated with the discovered one.  While this is indeed
  useful, especially in genetics where correlated SNPs tend to be very
  close to one another on the genome, this comes at a price since it
  significantly alters the meaning of the FDR. Indeed, if we adopt
  this viewpoint, the unit of inference is no longer a SNP but,
  rather, a region on the genome, yet FDR is still being tested at the
  SNP level.  A consequence of this mismatch is that the result of the
  analysis may be completely misleading, as beautifully argued in
  \citet{MP-ea:2004,YB-RH:2007,DS-NZ-BY:2011}, see also
  \citet{chouldechova2014false} and \citet{DB-ea:2016} for later
  references.
\item[3. Dependent p-values] Marginal p-values will, in general, have
  quite complicated joint dependence, so that BHq does not control FDR
  exactly. Although procedures for controlling FDR under arbitrary
  dependence exist, their increased generality tends to make them
  considerably more conservative than BHq. In practice, however, the
  FDR of BHq applied to dependent p-values is usually below its
  nominal level, but the problem is that it can have highly
  \emph{variable} FDP. Recall that FDR is the expectation of FDP, the
  latter being the random quantity we actually care about but cannot
  control directly. Therefore, FDR control is only useful if the
  realized FDP is relatively concentrated around its expectation, and
  it is well-established \citep[Chapter 4]{BE:2010} that under
  correlations BHq can produce highly skewed FDP distributions. In
  such cases, with large probability, $\text{FDP} =0$ perhaps because
  no discoveries are made, and when discoveries are made, the FDP may
  be much higher than the nominal FDR, making it a misleading error
  bound.
\end{description}

	\subsection{Getting valid p-values via conditional randomization
		testing}
	\label{sec:crtintro}
	If we insist on obtaining p-values for each $X_j$, there is in fact a
	simple method when the covariate distribution is assumed known, as it
	is in this paper. This method is similar in spirit to both propensity
	scoring (where the distribution of a binary $X_j$ conditional on the
	other variables is often estimated) and randomization/permutation
	tests (where $X_j$ is either the only covariate or fully independent
	of the other explanatory variables). Explicitly, a conditional
	randomization test for the $j$th variable proceeds by first computing
	some feature importance statistic $T_j$ for the $j$th variable. Then
	the null distribution of $T_j$ can be computed through simulation by
	independently sampling $X^*_j$'s from the \emph{conditional}
	distribution of $X_j$ given the others (derived from the known $F_X$)
	and recomputing the same statistic $T^*_j$ with each new $X^*_j$ in
	place of $X_j$, see \rev{Section~\ref{sec:CRT}} for details.  Despite its
	simplicity, we have not seen this test proposed previously in the
	literature, although it nearly matches the usual randomization test
	when the covariates are independent of one another.
	
	\subsection{Outline of the paper}
	\label{sec:outline}
	The remainder of the paper is structured as follows:
	\begin{itemize}
        \item Section~\ref{sec:problem_statement} frames the controlled
        selection problem in rigorous mathematical terms.
        \item Section~\ref{sec:methodology} introduces the \revv{model-X} knockoff
        procedure, examines its relationship with the earlier proposal
        of \citet{RB-EC:2015}, proposes knockoff constructions and
        feature statistics, and establishes FDR control.
        \item Section~\ref{sec:CRT} introduces the conditional randomization test.
        \item Section~\ref{sec:simulations} demonstrates through simulations
        that the {M\revv{X}} knockoff procedure controls the FDR in a number
        of settings where no other procedure does, and that when
        competitors exist, knockoffs is more powerful.
        \item \rev{Section~\ref{sec:robust} gives \revv{some preliminary} simulations using
          artificial and real data \revv{regarding} the robustness of M\revv{X}
          knockoffs to unknown/misspecified covariate distributions.}
        \item Section \ref{sec:realdata} applies our procedure to a
        case-control study of Crohn's disease in the United Kingdom.
        \item Section~\ref{sec:discussion} concludes the paper with
        extensions and potential lines of future research.
	\end{itemize}
	
	\newtheorem{theorem}{Theorem}[section]
	\newtheorem{lemma}[theorem]{Lemma}
	\newtheorem{corollary}[theorem]{Corollary}
	\newtheorem{proposition}[theorem]{Proposition}
	\newtheorem{definition}[theorem]{Definition}
	\newtheorem{indhyp*}{Induction Hypothesis}
	
	\providecommand{\tabularnewline}{\\}
	\floatstyle{ruled}
	\newfloat{algorithm}{tbp}{loa}
	\providecommand{\algorithmname}{Algorithm}
	\floatname{algorithm}{\protect\algorithmname}

	\newcommand{\noj}{{{\text{-}j}}}
	
	\newcommand{\fdp}{\textnormal{FDP}}
	\newcommand{\fdr}{\textnormal{FDR}}
	\newcommand{\sfdp}{\textnormal{FDP}_{\textnormal{dir}}}
	\newcommand{\sfdr}{\textnormal{FDR}_{\textnormal{dir}}}
	\newcommand{\mfdp}{\textnormal{mFDP}}
	\newcommand{\mfdr}{\textnormal{mFDR}}
	\newcommand{\msfdp}{\textnormal{mFDP}_{\textnormal{dir}}}
	\newcommand{\msfdr}{\textnormal{mFDR}_{\textnormal{dir}}}
	\newcommand{\Sh}{\hat{S}}
	
	\newcommand{\diag}{\operatorname{diag}}
	\newcommand{\swap}{\textnormal{swap}}

	\newcommand{\bXko}{\bXk}
	\newcommand{\XX}{[\bX, \, \bXk]}
	
	\def\eqd{\,{\buildrel d \over =}\,}
	\def\geqd{\,{\buildrel d \over \ge}\,}
	\def\leqd{\,{\buildrel d \over \le}\,}
	
	\section{Problem statement}
	\label{sec:problem_statement}
	
	
	To state the controlled variable selection problem carefully, suppose
	we have $n$ i.i.d.~samples from a population, each of the form
	$(X, Y)$, where $X = (X_1, \ldots, X_p) \in \R^p$ and $Y \in \R$. If the conditional distribution of $Y$ actually depends upon
	a smaller subset of these variables, we would like to classify each
	variable as relevant or not depending on whether it belongs to this
	subset or not. Mathematically speaking, we are looking for the Markov
	blanket $\mathcal{S}$, i.e.~the ``smallest'' subset $\mathcal{S}$ such
	that conditionally on $\{X_j\}_{j \in \mathcal{S}}$, $Y$ is
	independent of all other variables.
	For almost all joint
	distributions of $(X,Y)$, there exists a unique Markov blanket but
	there are pathological cases where it does not.  An example is this:
	suppose that $X_1$ and $X_2$ are independent Gaussian variables and
	that $X_3 = X_1 - X_2$.  Further assume that the distribution of $Y$
	depends upon the vector $X$ only through $X_1 + X_2$,
	e.g.,~$Y \, | X \sim \mathcal{N}(X_1 + X_2, 1)$.  Then the set of
	relevant variables---or equivalently, the Markov blanket---is ill
	defined since we can say that the likelihood of $Y$ depends upon $X$
	through either $(X_1, X_2)$, $(X_1, X_3)$, or $(X_2, X_3)$, all these
	subsets being equally good. In order to define a unique set of
	relevant variables, we shall work with the notion of conditional
	\emph{pairwise} independence.
	\begin{definition}
		\label{def:null}
		A variable $X_j$ is said to be ``null'' if and only if $Y$ is
		independent of $X_j$ conditionally on the other variables
		$X_\noj = \{X_1, \ldots X_p\}\setminus \{X_j\}$.  The subset of null
		variables is denoted by $\mathcal{H}_0 \subset \{1, \ldots p\}$ and
		we call a variable $X_j$ ``nonnull'' or relevant if
		$j \notin \mathcal{H}_0$.
	\end{definition}
	From now on, {\em our goal is to discover as many relevant
		(conditionally dependent) variables as possible while keeping the
		FDR under control}.\footnote{Using the methods of
		\citet{LJ-WS:2016}, other error rates such as the $k$-familywise
		error rate can also be controlled using {M\revv{X}} knockoffs, but we
		focus on FDR for this paper.} Formally, for a selection rule that
	selects a subset $\hat{\mathcal{S}}$ of the covariates,
	\begin{equation}
		\label{eq:fdr2}
		\text{FDR} := \mathbb{E} \,\, \left[
		\frac{|\hat{\mathcal{S}} \cap \mathcal{H}_0|}{|\hat{\mathcal{S}}|} \right].
	\end{equation}
	
	In the example above, because of the perfect functional relationship
	$X_3 = X_2 - X_1$, all three variables $X_1, X_2, X_3$ would be
	classified as nulls. Imagine, however, breaking this relationship by
	adding a bit of noise, e.g.,~$X_3 = X_2 - X_1 + Z$, where $Z$ is
	Gaussian noise (independent of $X_1$ and $X_2$) however small. Then
	according to our definition, $X_1$ and $X_2$ are both nonnull while
	$X_3$ is null---and everything makes sense. Having said this, we
	should not let ourselves be distracted by such subtleties. In the
	literature on graphical models there, in fact, exist weak regularity
	conditions that guarantee that the (unique) set of relevant variables
	defined by pairwise conditional independence, exactly coincides with
	the Markov blanket so that there is no ambiguity.  In this field,
	researchers typically assume these weak regularity conditions hold
	(examples would include the local and global Markov properties), and
	proceed from there. For example, the textbook \citet{DE:2000}
	describes these properties on page 8 as holding ``under quite general
	conditions'' and then assumes them for the rest of the book.

	Our definition is very natural to anyone working with parametric
	GLMs. In a GLM, the response $Y$ has a
	probability distribution taken from an exponential family, which
	depends upon the covariates only through the linear combination
	$\eta = \beta_1 X_1 + \cdots + \beta_p X_p$. The relationship between
	$Y$ and $X$ is specified via a link function $g$ such that
	$\E{Y | X} = g^{-1}(\eta)$. In such models and under broad conditions,
	$Y \independent X_j \, | \, X_{\text{-}j}$ if and only if
	$\beta_j = 0$. In this context, testing the hypothesis that $X_j$ is a
	null variable is the same as testing $H_j: \, \beta_j = 0$.
	\begin{proposition}
		Take a family of random variables $X_1, \ldots, X_p$ such that one
		cannot perfectly predict any one of them from knowledge of the
		others. If the likelihood of $Y$ follows a GLM, then
		$Y \independent X_j \, | \, X_{\text{-}j}$ if and only if
		$\beta_j = 0$. Hence, $\mathcal{H}_0$ from Definition \ref{def:null}
		is exactly the set $\{j : \beta_j = 0\}$.
	\end{proposition}
	\begin{proof}
		We prove this in the case of the logistic regression model as the
		general case is similar. Here, the conditional distribution of $Y$ is
		Bernoulli with
		\[
		\E{Y | X} = \mathbb{P}(Y = 1 | X) = \frac{e^\eta}{1 + e^{\eta}} =
		g^{-1}(\eta), \quad \eta = \beta_1 X_1 + \cdots + \beta_p X_p,
		\]
		and please note that the assumption about the covariates implies that
		the model is identifiable. Now assume first that $\beta_j = 0$. Then
		\begin{align}
			p_{Y, X_j | X_\noj}(y, x_j \, | \, x_\noj)
			& =
			p_{Y|X_j, X_\noj}
			(y \, | \, x_j, x_\noj) \, p_{X_j|X_\noj}(x_j \, | \, x_\noj)
		\end{align}
		and since the first factor in the right-hand side does not depend on
		$X_j$, we see that the conditional probability distribution function
		factorizes. This implies conditional independence.  In the other
		direction, assume that $Y$ and $X_j$ are conditionally
		independent. Then the likelihood function
		\[
		\frac{\exp(Y {(\beta_1 X_1+\cdots+\beta_pX_p)})}{1 + \exp({\beta_1 X_1+\cdots+\beta_pX_p})}
		\]
		must, conditionally on $X_\noj$, factorize into a function of $Y$
		times a function of $X_j$. A consequence of this is that conditionally
		on $X_\noj$, the odds ratio must not depend on $X_j$ (it must be
		constant). However, this ratio is equal to $\exp(\beta_j X_j)$ and is
		constant only if $\beta_j = 0$ since, by assumption, $X_j$ is not
		determined by $X_\noj$.
	\end{proof}
	
	The assumption regarding the covariates is needed. Indeed, suppose
	$X_1 \sim \mathcal{N}(0,1)$, $X_2 = 1\{X_1 > 0\}$, and $Y$ follows a
	logistic model as above with $\eta = X_1 + X_2$. Then
	$Y \independent X_2 \, | \, X_1$ even though $\beta_2 = 1$. In this
	example, the conditional distribution of $Y$ depends on $(X_1, X_2)$
	only through $X_1$. Therefore, for the purpose of identifying
	important variables (recall our task is to find important variables
	and not to learn exactly how the likelihood function depends upon these
	variables), we would like to find $X_1$ and actually do not care about
	$X_2$ since it provides no new information.

	

	
	\section{Methodology}
	\label{sec:methodology}

	\subsection{Model-\revv{X} knockoffs}
	\label{sec:MF_knockoffs}
	
	\subsubsection{Definition}
	\label{sec:def_MFk}
	\begin{definition} {\bf Model-\revv{X} knockoffs} for the family of random
		variables $X = (X_1, \ldots, X_p)$ are a new family of random
		variables $\tilde X = (\tilde X_1, \ldots, \tilde X_p)$ constructed
		with the following two properties: (1) for any subset
		$S \subset \{1, \ldots, p\}$,\footnote{$\eqd$ denotes equality in
			distribution, and the definition of the
			swapping operation is given just below.}
		\begin{equation}
			\label{eq:randomko}
			(X, \,  \tilde{X})_{\swap{(S)}} \, \eqd \, (X, \, \tilde
			X);
		\end{equation}
		(2) $\tilde X \independent Y \, | \, X$ if there is a response $Y$.
		(2) is guaranteed if $\tilde X$ is constructed without looking at $Y$.
		\label{def:def_MfK}
	\end{definition}
	Above, the vector $(X, \, \tilde{X})_{\swap{(S)}}$ is obtained from
	$(X, \, \tilde{X})$ by swapping the entries $X_j$ and $\tilde{X}_j$
	for each $j\in S$; for example, with $p = 3$ and $S = \{2,3\}$,
	\[
	(X_1, X_2, X_3, \tilde X_1, \tilde X_2, \tilde X_3) _{\swap{(\{2,3\})}} \, \rev{\eqd}
	\,  (X_1, \tilde X_2,
	\tilde X_3, \tilde X_1, X_2, X_3).
	\]
	We see from \eqref{eq:randomko} that original and knockoff variables
	are pairwise exchangeable: taking any subset of variables and
	swapping them with their knockoffs leaves the joint distribution
	invariant. \rev{Note that our exchangeability condition is on
          the covariates, and thus bears little resemblance to
          exchangeability conditions for closed permutation testing
          (see, e.g., \citet{PW-JT:2008}).} To give an example of {M\revv{X}} knockoffs, suppose that
	$X \sim \mathcal{N}(0, \bS)$. Then a joint distribution obeying
	\eqref{eq:randomko} is this:
	\begin{equation}
		\label{eq:ko_Gaussian}
		(X, \, \tilde X) \sim \mathcal{N}(0, \bG), \quad \text{where} \quad  \bG = \begin{bmatrix} \bS & \bS - \diag\{s\}\\
			\bS - \diag\{s\} & \bS \end{bmatrix};
	\end{equation}
	here, $\diag\{s\}$ is any diagonal matrix selected in such a
        way that the joint covariance matrix $\bG$ is positive
        semidefinite.  Indeed, the distribution obtained by swapping
        variables with their knockoffs is Gaussian with a covariance
        given by $\bs{P} \bG \bs{P}$, where $\bs{P}$ is the
        permutation matrix encoding the swap. Since
        $\bs{P} \bG \bs{P} = \bG$ for any swapping operation, the
        distribution is invariant. For an interesting connection with
        the invariance condition in \citet{RB-EC:2015}, see
        Appendix~\ref{app:detko}.
	
	We will soon be interested in the problem of constructing knockoff
	variables, having observed $X$. In the above example, a possibility is
	to sample the knockoff vector $\tilde X$ from the conditional
	distribution
	\[
	\tilde X \, | X \eqd \mathcal{N}(\mu, \bs V),
	\]
	where $\mu$ and $\bs V$ are given by classical regression formulas,
	namely,
	\begin{align*}
		\mu & = X - X\rev{\bS^{-1}\diag\{s\}}, \\
		\bs V & = 2\diag\{s\} - \diag\{s\}\bS^{-1} \diag\{s\}.
	\end{align*}
	There are, of course, many other ways of constructing knockoff
	variables, and for the time being, we prefer to postpone the
	discussion of more general constructions.
	
	In the setting of the paper, we are given i.i.d.~pairs
	$(X_{i1}, \ldots, X_{ip}, Y_i) \in \R^p \times \R$ of covariates and
	responses, which we can assemble in a data matrix $\bX$ and a data
	vector $y$ in such a way that the $i$th row of $\bX$ is
	$(X_{i1}, \ldots, X_{ip})$ and the $i$th entry of $y$ is $Y_i$. Then
	the {M\revv{X}} knockoff matrix $\tilde{\bX}$ is constructed in such a
	way that for each observation label $i$,
	$(\tilde X_{i1}, \ldots, \tilde{X}_{ip})$ is a knockoff for
	$(X_{i1}, \ldots, X_{ip})$ as explained above; that is to say, the
	joint vector
	$(X_{i1}, \ldots, X_{ip}, \tilde{X}_{i1}, \ldots, \tilde{X}_{ip})$
	obeys the pairwise exchangeability property \eqref{eq:randomko}.
	
	\subsubsection{Exchangeability of null covariates and their knockoffs}
	
	A crucial property of {M\revv{X}} knockoffs is that we can swap null
	covariates with their knockoffs without changing the joint
	distribution of the original covariates $X$ \rev{and} their knockoffs
	$\tilde X$, \rev{{\em conditional} on} the response $Y$.  From now on, $X_{i:j}$ for
	$i \le j$ is a shorthand for $(X_i, \ldots, X_{j})$.
	\begin{lemma}
		\label{lem:exch}
		Take any subset $S \subset \mathcal{H}_0$ of nulls. Then
		\[
                  \rev{\XX\mid y \, \eqd \, \XX_{\swap{(S)}}\mid y.}
		\]
	\end{lemma}
        \begin{proof}
          \rev{Since $y$'s marginal distribution is the same on both
            sides of the equation, it is equivalent to show that
            $(\XX, \, y) \, \eqd \, (\XX_{\swap{(S)}}, \, y)$, which
            is how we proceed.}  Assume without loss of generality
          that $S = \{1, 2, \ldots, m\}$.  By row independence, it
          suffices to show that $((X, \tilde X), \, Y)$ $\eqd$
          $((X, \tilde X)_{\swap{(S)}}, \, Y)$, where $X$ (resp.~$Y$)
          is a row of $\bX$ (resp.~$y$).  Furthermore, since
          $(X, \tilde X)$ $\eqd$ $(X, \tilde X)_{\swap{(S)}}$, we only
          need to establish that
		\begin{equation}
			\label{eq:toprove}
			Y \, | \, (X, \tilde X)_{\swap{(S)}} \, \eqd \, Y \, | \, (X, \tilde
			X).
		\end{equation}
		Letting $p_{Y  |  X}(y  |  x)$ be the conditional
		distribution of $Y$, observe that
		\begin{align*}
			p_{Y  |  (X, \tilde X)_{\swap{(S)}}} (y  |  (x, \tilde x)) & = p_{Y  |  (X, \tilde X)} (y  |  (x, \tilde x)_{\swap{(S)}})\\
			& = p_{Y | X} (y  |  x'),
		\end{align*}
		where $x'_i = \tilde{x}_i$ if $i \in S$ and $x'_i = x_i$
		otherwise. The second equality above comes from the fact that $Y$ is
		conditionally independent of $\tilde X$ by property (2) in the
		definition of M\revv{X} knockoffs. Next, since $Y$ and $X_1$ are independent
		conditional on $X_{2:p}$, we have
		\begin{align*}
			p_{Y | X_{1:p}} (y  |  \tilde x_1, x'_{2:p}) & = p_{Y | X_{2:p}} (y  |  x'_{2:p}) \\ & = p_{Y | X_{1:p}} (y  |  x_1, x'_{2:p}).
		\end{align*}
		This shows that
		\[
		Y \, | \, (X, \tilde X)_{\swap{(S)}} \, \eqd \, Y \, | \, (X, \tilde X)_{\swap{(S\setminus\{1\})}}.
		\]
		We can repeat this argument with the second variable, the third, and so on  until $S$ is empty. This proves \eqref{eq:toprove}.
        \end{proof}
	
	\subsection{Feature statistics}
	\label{sec:W}
	
	In order to find the relevant variables, we now compute statistics
	$W_j$ for each $j \in \{1, \ldots, p\}$, a large positive value of
	$W_j$ providing evidence against the hypothesis that $X_j$ is null.
	This statistic depends on the response and the original variables
	but also on the knockoffs; that is,
	\[
	W_j = w_j([\bX, \, \tilde{\bX}], \, y)
	\]
	for some function $w_j$.  As in \citet{RB-EC:2015}, we impose a {\em
		flip-sign property}, which says that swapping the $j$th variable
	with its knockoff has the effect of changing the sign of
	$W_j$. Formally, if $[\bX, \, \tilde{\bX}]_{\swap(S)}$ is the matrix
	obtained by swapping columns in $S$,
	\begin{equation}
		\label{eq:antisymmetry}
		w_j([\bX, \, \tilde{\bX}]_{\swap(S)}, \, y) =
		\begin{cases} w_j([\bX, \, \tilde{\bX}],y), &j\not\in S,\\  - w_j([\bX, \, \tilde{\bX}],y), & j\in S.\end{cases}
	\end{equation}
	In contrast to the aforementioned work, we do not require the sufficiency
	property that $w_j$ depend on $\bX$, $\tilde{\bX}$, and $y$ only
	through $[\bX, \, \tilde{\bX}]^{\top}[\bX, \, \tilde{\bX}]$ and
	$[\bX, \, \tilde{\bX}]^{\top}y$.
	
	At this point, it may help the reader unfamiliar with the knockoff
	framework to think about knockoff statistics $W = (W_1, \ldots, W_p)$
	in two steps: first, consider a statistic $T$ for each original and
	knockoff variable,
	\[
	T \triangleq (Z, \, \tilde Z) = (Z_1, \ldots, Z_p, \tilde Z_1, \ldots,
	\tilde Z_p) = t([\bX, \, \bXko], \, y),
	\]
	with the idea that $Z_j$ (resp.~$\tilde{Z}_j$) measures the importance
	of $X_j$ (resp.~$\tilde{X}_j$). Assume the natural property that
	{switching} a variable with its knockoff simply {switches} the
	components of $T$ in the same
	way, 
	namely, for each $S \subset \{1, \ldots, p\}$,
	\begin{equation}
		\label{eq:obvious}
		(Z, \, \tilde Z)_{\swap{(S)}}  =  t([\bX, \, \bXko]_{\swap{(S)}}, \, y).
	\end{equation}
	Then one can construct a $W_j$ obeying the flip-sign condition
	\eqref{eq:antisymmetry} by setting
	\[
	W_j = f_j(Z_j, \tilde{Z}_j),
	\]
	where $f_j$ is any anti-symmetric function.\footnote{An anti-symmetric
		function $f$ is such that $f(v,u) = - f(u,v)$.}  (Conversely, any
	statistic $W_j$ verifying the flip sign condition can be constructed
	in this fashion.) Adopting this approach, consider a regression
	problem and run the Lasso on the original design augmented with
	knockoffs,
	\begin{equation}
		\label{eq:augmentedLasso}
		\text{min}_{b \in \R^{2p}} \quad \textstyle{\frac12} \|y - [\bX, \,
		\bXk]
		b\|_2^2 + \lambda \|b\|_1
	\end{equation}
	and denote the solution by $\hat{b}(\lambda)$ (the first $p$
	components are the coefficients of the original variables and the last
	$p$ are for the knockoffs). Then the {\em Lasso
		coefficient-difference} (LCD) statistic sets
	$Z_j = |\hat{b}_j(\lambda)|$, $\tilde{Z}_j = |\hat{b}_{j+p}(\lambda)|$, and
	\begin{equation}
		\label{eq:LCD}
		W_j = Z_j - \tilde{Z}_j = |\hat{b}_j(\lambda)| -
		|\hat{b}_{j+p}(\lambda)|.
	\end{equation}
	A large positive value of $W_j$ provides some evidence that the
	distribution of $Y$ depends upon $X_j$, whereas under the null $W_j$
	has a symmetric distribution and, therefore, is equally likely to take on
	positive and negative values, as we shall see next. Before moving on,
	however, please carefully observe that the value of $\lambda$ in
	\eqref{eq:LCD} does not need to be fixed in advance, and can be
	computed from $y$ and $\XX$ in any data-dependent fashion as long as
	permuting the columns of $\bX$ does not change its value; for
	instance, it can be selected by cross-validation.
	
	\begin{lemma}
		\label{lem:key}
		Conditional on \rev{$(|W_1|,\dots,|W_p|)$}, the signs of the null $W_j$'s,
		$j \in \mathcal{H}_0$, are i.i.d.~coin flips.
	\end{lemma}
	\begin{proof}
		Let $\epsilon = (\epsilon_1, \ldots, \epsilon_p)$ be a sequence of
		independent random variables such that $\epsilon_j = \pm 1$ with
		probability $1/2$ if $j \in \mathcal{H}_0$, and $\epsilon_j = 1$
		otherwise.  To prove the claim, it suffices to establish that
		\begin{equation}
			\label{eq:key}
			W \, \eqd  \, \epsilon \odot W,
		\end{equation}
		where $\odot$ denotes pointwise multiplication, i.e.
		$\epsilon \odot W = (\epsilon_1 W_1, \ldots, \epsilon_p W_p)$.  Now,
		take $\epsilon$ as above and put
		$S = \{j : \epsilon_j = -1\} \subset \mathcal{H}_0$.  Consider
		swapping variables in $S$:
		\[
		W_{\swap(S)} \triangleq w (\XX_{\swap(S)} , y).
		\]
		On the one hand, it follows from the flip-sign property that
		$W_{\swap(S)} = \epsilon \odot W$. On the other hand, Lemma
		\ref{lem:exch} implies that $W_{\swap(S)} \, \eqd \, W$ since
		$S \subset \mathcal{H}_0$. These last two properties give \eqref{eq:key}.
	\end{proof}
        \rev{In fact, since the pairwise exchangeability property of $\XX$
          holds conditionally on $y$ according to Lemma~\ref{lem:exch},
          the coin-flipping property also holds conditionally on $y$.}
	
	
	\subsection{FDR control}
	\label{sec:knockoffs}
	
	From now on, our methodology follows that of \citet{RB-EC:2015} and we
	simply rehearse the main ingredients while referring to their paper
	for additional insights. It follows from Lemma \ref{lem:key} that the
	null statistics $W_j$ are symmetric and that for any fixed threshold
	$t > 0$,
	\[
	\#\{j:W_j\leq -t\} \geq \#\{\text{null }j:W_j\leq -t\}\eqd
	\#\{\text{null }j:W_j\geq t\}.
	\]
	Imagine then selecting those variables such that $W_j$ is sufficiently
	large, e.g.,~$W_j\ge t$, then the false discovery proportion (FDP)
	\begin{equation}
		\label{eq:fdp}
		\fdp(t) = \frac{\# \{\text{null }j:W_j\geq t\}}{\# \{ j: W_j\geq t\}}
	\end{equation}
	can be estimated via the statistic
	\[
	\widehat{\fdp}(t) = \frac{\#\{j:W_j\leq -t\}}{\# \{ j: W_j\geq t\}}
	\]
	since the numerator is an upward-biased estimate of the unknown
	numerator in \eqref{eq:fdp}. The idea of the knockoff procedure is to
	choose a data-dependent threshold as liberal as possible while having
	an estimate of the FDP under control. The theorem below shows that
	estimates of the FDR process can be inverted to give tight FDR
	control.
	\begin{theorem}
		\label{thm:main}
		Choose a threshold $\tau>0$ by setting\footnote{When we write
			$\min\{t > 0 : \ldots\}$, we abuse notation since we actually mean
			$\min\{t \in \mathcal{W}_+ : \ldots\}$, where
			$\mathcal{W}_+ = \{|W_j|: |W_j| > 0\}$.}
		\begin{equation}\label{eqn:knockoff_stoppingtime}
			\tau = \min\left\{t > 0 : \frac{\#\{j: W_j \leq - t\}}{\#\{j: W_j\geq
				t\}} \leq  q \right\} \quad (\textbf{Knockoffs}),
		\end{equation}
		where $q$ is the target FDR level (or $\tau = +\infty$ if the set above
		is empty).  Then the procedure selecting the variables
		\[\Sh = \{j : W_j \geq \tau\}\;\] controls the
		modified FDR defined as
		\[
		\mfdr = \mathbb{E} \left[ \frac{|\{j \in \hat S \cap \mathcal{H}_0\}|}{|\hat
			S| + 1/q} \right] \le q.
		\]
		The slightly more conservative procedure, given by incrementing the
		number of negatives by one,
		\begin{equation}
			\label{eqn:knockoffplus_stoppingtime}
			\tau_+ = \min\left\{t > 0 : \frac{1+\#\{j: W_j \leq - t\}}{\#\{j:
				W_j\geq t\}} \leq q \right\}  \quad (\textbf{Knockoffs+})
		\end{equation}
		and setting $\Sh = \{j: W_j\geq \tau_+\}$, controls the usual FDR,
		\[
		\mathbb{E} \left[ \frac{|\{j \in \hat S \cap \mathcal{H}_0\}|}{|\hat
			S| \vee 1} \right] \le q.
		\]
		These results are non-asymptotic and hold no matter the dependence
		between the response and the covariates\rev{---in
                  fact, they hold \emph{conditionally} on the response $y$}.
	\end{theorem}
	The proof is the same as that of Theorems 1 and 2 in
        \citet{RB-EC:2015}---and, therefore, omitted---since all we
        need is that the null statistics have signs distributed as
        i.i.d.~coin flips \rev{(even conditionally on $y$)}. Note that
        Theorem~\ref{thm:main} only tells one side of the story: Type
        I error control; the other very important side is power, which
        leads us to spend most of the remainder of the paper
        considering how best to construct knockoff variables and
        statistics.

	\subsection{Constructing model-\revv{X} knockoffs}
	\label{sec:construction}
	
	\subsubsection{Exact constructions}
	
	We have seen in Section~\ref{sec:def_MFk} one way of constructing
	{M\revv{X}} knockoffs in the case where the covariates are Gaussian.
	How should we proceed for non-Gaussian data? In this regard, the
	characterization below may be useful.
	\begin{proposition}
		\label{prop:characterization}
		The random variables $(\tilde{X}_1, \ldots, \tilde{X}_p)$ are
		model-\revv{X} knockoffs for $(X_1, \ldots, X_p)$ if and only if for any
		$j \in \{1, \ldots, p\}$, the pair $(X_j, \tilde{X}_j)$ is
		exchangeable conditional on all the other variables and their
		knockoffs (and, of course, $\tilde X \independent Y \, | \, X$).
	\end{proposition}
	The proof consists of simple manipulations of the definition and is,
	therefore, omitted. Our problem can thus also be posed as constructing pairs that are
	conditionally exchangeable. If the components of the vector $X$ are
	independent, then any independent copy of $X$ would work; that is, any
	vector $\tilde X$ independently sampled from the same joint
	distribution as $X$ would work.  With dependent coordinates, we may
	proceed as follows:
	\begin{algorithm}[H]
		\caption{Sequential Conditional Independent Pairs.\label{alg:sequential}}
		\SetAlgoLined\DontPrintSemicolon
		{$j = 1$ \;
			\While{$j  \leq p$} {
				Sample $\tilde{X}_j$ from $\mathcal{L}(X_j \, | \, X_\noj, \, \tilde{X}_{1:j-1})$\; \\
				$j = j + 1$\;
			}
		}
	\end{algorithm}
	Above, $\mathcal{L}(X_j \, | \, X_\noj, \, \tilde{X}_{1:j-1})$ is the
	conditional distribution of $X_j$ given $(X_\noj, \tilde{X}_{1:j-1})$.
	When $\rev{p} = 3$, this would work as follows: sample $\tilde{X}_1$ from
	$\mathcal{L}(X_1 \, | \, X_{2:3})$. Once this is done,
	$\mathcal{L}(X_{1:3}, \tilde{X}_1)$ is available and we, therefore,
	know $\mathcal{L}(X_2 \, | \, X_1, X_3, \tilde{X}_1)$.  Hence, we can
	sample $\tilde{X}_2$ from this distribution. Continuing,
	$\mathcal{L}(X_{1:3}, \tilde{X}_{1:2})$ becomes known and we can
	sample $\tilde{X}_3$ from
	$\mathcal{L}(X_3 \, | \, X_{1:2}, \tilde{X}_{1:2})$.
	
	It is not immediately clear why Algorithm \ref{alg:sequential} yields
	a sequence of random variables obeying the exchangeability property
	\eqref{eq:randomko}, and we prove this fact in
	Appendix~\ref{app:scip}.
	There is, of course, nothing special about the ordering in which
	knockoffs are created and equally valid constructions may be
	obtained by looping through an arbitrary ordering of the
	variables. For example, in a data analysis application where we
	would need to build a knockoff copy for each row of the design,
	independent (random) orderings may be used.
	
	To have power or, equivalently, to have a low Type II error rate, it
	is intuitive that we would like to have original features $X_j$ and
	their knockoff companions $\tilde{X}_j$ to be as ``independent'' as
	possible.
	
	We do not mean to imply that running Algorithm \ref{alg:sequential} is
	a simple matter. In fact, it may prove rather complicated since we
	would have to recompute the conditional distribution at each step;
	this problem is left for future research. Instead, in this paper we
	shall work with approximate {M\revv{X}} knockoffs and will demonstrate
	empirically that for models of interest, such constructions yield FDR
	control.
	
	\subsubsection{Approximate constructions: second-order model-\revv{X}
		knockoffs}
	\label{sec:asdp}
	Rather than asking that $(X, \, \tilde{X})_{\swap{(S)}}$ and
	$(X, \, \tilde X)$ have the same distribution for any subset $S$, we
	can ask that they have the same first two moments, i.e., the same mean
	and covariance. Equality of means is a simple matter. As far as the
	covariances are concerned, equality is equivalent to
	\begin{equation}
		\label{eq:ko_cov}
		\operatorname{cov}(X, \, \tilde X) = \bG, \quad \text{where} \quad  \bG = \begin{bmatrix} \bS & \bS - \diag\{s\}\\
			\bS - \diag\{s\} & \bS \end{bmatrix}.
	\end{equation}
	We, of course, recognize the same form as in \eqref{eq:ko_Gaussian}
	where the parameter $s$ is chosen to yield a positive semidefinite
	covariance matrix. (When $(X, \tilde X)$ is Gaussian, a matching of
	the first two moments implies a matching of the joint distributions so
	that we have an exact construction.)  Furthermore,
	Appendix~\ref{app:detko} shows that the same problem was already solved
	in \citet{RB-EC:2015}, as the same constraint on $s$ applies but with
	the empirical covariance replacing the true covariance. This means
	that the same two constructions proposed in \citet{RB-EC:2015} are
	just as applicable to {\em second-order model-\revv{X} knockoffs}.
	
	For the remainder of this section, we will assume the covariates have
	each been translated and rescaled to have mean zero and variance
	one. To review, the \emph{equicorrelated} construction uses
	\[s^{\text{EQ}}_j = 2\lambda_{\text{min}}(\bS)\wedge 1 \text{ for
		all }j,\]
	which minimizes the correlation between variable-knockoff pairs
	subject to the constraint that all such pairs must have the same
	correlation. The \emph{semidefinite program (SDP)} construction solves
	the convex program
	\begin{equation}
		\label{eq:sdp}
		\begin{array}{rl}
			\text{minimize} & \quad \sum_{j}|1-s^{\text{SDP}}_j| \\
			\text{subject to} & \quad s^{\text{SDP}}_j \ge 0 \\
			& \quad \text{diag}\{s^{\text{SDP}}\}\preceq 2\bS, \end{array}
	\end{equation}
	which minimizes the sum of absolute values of variable-knockoff
	correlations among all suitable $s$.
	
	In applying these constructions
	to problems with large $p$, we run into some new difficulties:
	\begin{itemize}
		\item Excepting very specially-structured matrices like the identity,
		$\lambda_{\text{min}}(\bS)$ tends to be extremely small as $p$
		gets large. The result is that constructing equicorrelated
		knockoffs in high dimensions, while fairly computationally easy,
		will result in very low power, since all the original variables will
		be nearly indistinguishable from their knockoff counterparts.
		\item For large $p$, \eqref{eq:sdp}, while convex, is prohibitively
		computationally expensive. However, if it could be computed, it
		would produce much larger $s_j$'s than the equicorrelated
		construction and thus be considerably more powerful.
	\end{itemize}
	To address these difficulties, we first generalize the two knockoff
	constructions by the following two-step procedure, which we call
	the approximate semidefinite program (ASDP) construction:
	\begin{description}
		\item[Step 1.] Choose an approximation $\bS_{\text{approx}}$ of
		$\bS$ and solve:
		\begin{equation}
			\label{eq:asdp}
			\begin{array}{rl}
				\text{minimize} & \quad \sum_{j}|1-\hat{s}_j| \\
				\text{subject to} & \quad \hat{s}_j \ge 0 \\
				& \quad \text{diag}\{\hat{s}\}\preceq 2\bS_{\text{approx}}. \end{array}
		\end{equation}
		\item[Step 2.] Solve:
		\begin{equation}
			\begin{array}{rl}
				\text{maximize} & \quad \gamma \\
				\text{subject to} & \quad \text{diag}\{\gamma\hat{s}\}\preceq 2\bS, \end{array}
		\end{equation}
		and set $s^{\text{ASDP}} = \gamma \hat{s}$. Note this problem can be solved quickly
		by, e.g., bisection search over $\gamma\in [0,1]$.
	\end{description}
	ASDP with $\bS_{\text{approx}} = \bs I$ trivially gives $\hat{s}_j=1$
	and $\gamma = 2\lambda_{\text{min}}(\bS)\wedge 1$, reproducing the
	equicorrelated construction. ASDP with $\bS_{\text{approx}}=\bS$
	clearly gives $\hat{s}_j=s^{\text{SDP}}$ and $\gamma=1$, reproducing
	the SDP construction. Note that the ASDP step 2 is always fast, so the
	speed of the equicorrelated construction comes largely because the
	problem \emph{separates} into $p$ computationally independent SDP
	subproblems of
	$\text{min}|1-\hat{s}_j|\text{ s.t. } 0\le \hat{s}_j\le 2$. However,
	power is lost due to the very na\"{i}ve approximation
	$\bS_{\text{approx}}= \bs I$ which results in a very small $\gamma$.
	
	In general, we can choose $\bS_{\text{approx}}$ to be an
	$m$-block-diagonal approximation of $\bS$, so that the ASDP from
	Step 1 separates into $m$ smaller, more computationally tractable, and
	trivially parallelizable SDP subproblems. If the approximation is
	fairly accurate, we may also find that $\gamma$ remains large, so that
	the knockoffs are nearly as powerful as if we had used the SDP
	construction. We demonstrate the ASDP construction in
	Section~\ref{sec:realdata} when we analyze the Crohn's disease data.
	
	\section{The conditional randomization test}
	\label{sec:CRT}
	
	This section presents an alternative approach to the controlled
	variable selection problem.  To describe our approach, it may be best
	to consider an example. Assume we are in a regression setting and let
	$\hat{b}_j(\lambda)$ be the value of the Lasso estimate of the $j$th
	regression coefficient. We would like to use the statistic
	$\hat{b}_j(\lambda)$ to test whether $Y$ is conditionally independent
	of $X_j$ since large values of $|\hat{b}_j(\lambda)|$ provide evidence
	against the null. To construct a test, however, we would need to know
	the sampling distribution of $\hat{b}_j(\lambda)$ under the null
	hypothesis that $Y$ and $X_j$ are conditionally independent, and it is
	quite unclear how one would obtain such knowledge.
	
	\subsection{The test}
	
	A way out is to sample the covariate $X_j$ conditional on all the
	other covariates (but not the response), where by ``sample'' we
	explicitly mean to draw a new sample from the conditional distribution
	of $X_j\,|\,X_{\text{-}j}$ using a random number generator. We then
	compute the Lasso statistic $\hat{b}^*_j(\lambda)$, where the
	$^*$ superscript indicates that the statistic is computed from the artificially
	sampled value of the covariate $X_j$. Now, under the null hypothesis
	of conditional independence between $Y$ and $X_j$, it happens that
	$\hat{b}^*_j(\lambda)$ and $\hat{b}_j(\lambda)$ are identically
	distributed and that, furthermore, this statement holds true
	conditional on $Y$ and all the other covariates. This claim is proved
	in Lemma \ref{lem:CRT} below. A consequence of this is that by
	simulating a covariate conditional on the others, we can sample at
	will from the conditional distribution of any test statistic and compute
	p-values as described in Algorithm \ref{alg:CRT}.
	
	\begin{algorithm*}[t]
		\caption{Conditional Randomization Test.\label{alg:CRTa}}
		\label{alg:CRT}
		\begin{tabular}{>{\raggedright}p{1\textwidth}}
			\textbf{Input}:
			A set of $n$ independent samples $(X_{i1}, \ldots, X_{ip}, Y_i)_{1
				\le i \le n}$ assembled in a data
			matrix $\bX$ and a response vector $y$, a feature importance
			statistic
			$T_j(\bX, y)$  to test whether $X_j$ and $Y$
			are conditionally independent.
			\vspace{0.7em}
			\tabularnewline
			\textbf{Loop: for $k=1,2,\ldots,K$ do} \\
			Create a new data matrix $\bX^{(k)}$ by simulating the $j$th
			column of $\bX$ from $\mathcal{L}(X_j|X_{\text{-}j})$ (and keeping
			the remaining columns the same). That is,
			${X}_{ij}^{(k)}$ is sampled from the conditional
			distribution $X_{ij} \, | \, \{X_{i1}, \ldots, X_{ip}\}\setminus \{X_{ij}\}$, and is (conditionally)
			independent of $X_{ij}$.  \vspace{0.7em}
			\tabularnewline
			\textbf{Output:} A (one-sided) p-value
			\[
			P_j = \frac{1}{K+1} \left[ 1 + \sum_{k = 1}^K \one{T_j(\bX^{(k)},y) \ge
				T_j({\bX} ,y)} \right].
			\]
			As with permutation
			tests, adding one in both the numerator and the denominator
			makes sure that the
			null  p-values are stochastically larger than uniform variables.
			\tabularnewline
		\end{tabular}
	\end{algorithm*}
	
	\begin{lemma}
		\label{lem:CRT} Let $(Z_1, Z_2, Y)$ be a triple of random variables, and
		construct another triple $(Z_1^*, Z_2, Y)$ as
		\[
		Z_1^* \, | \, (Z_2,Y) \quad \eqd \quad Z_1 \, | \, Z_2.
		\]
		Then under the null hypothesis $Y \independent Z_1 \, | \, Z_2$, any test
		statistic $T = t(Z_1, Z_2, Y)$ obeys
		\[
		T \, | \, (Z_2, Y) \quad \eqd \quad T^* \, | \, (Z_2,Y),
		\]
		where $T^* = t(Z_1^*, Z_2, Y)$.
	\end{lemma}
	\begin{proof}
		To prove the claim, it suffices to show that $Z_1$ and $Z_1^*$ have
		the same distribution conditionally on $(Z_2, Y)$. This follows from
		\begin{align*}
			Z_1^* \, | \, (Y, Z_2) \quad \eqd \quad Z_1 \, | \, Z_2 \quad \eqd \quad Z_1 \, |
			\, (Z_2, Y).
		\end{align*}
		The first equality comes from the definition of $Z_1^*$ while the
		second follows from the conditional independence of $Y$ and $Z_1$,
		which holds under the null.
	\end{proof}
	The consequence of Lemma \ref{lem:CRT} is that we can compute the 95\%
	percentile, say, of the conditional distribution of $T^*$ denoted
	by $t^*_{0.95}(Z_2,Y)$. Then by definition, under the null,
	\[
	\mathbb{P}(T > t^*_{0.95}(Z_2,Y) \, | \, (Z_2,Y)) \le 0.05.
	\]
	Since this equality holds conditionally, it also holds marginally.
	
	\subsection{Literature review}
	
	The conditional randomization test is most closely related to the propensity score
	\citep{PR-DR:1983}, which also uses the conditional distribution
	$X_j \,| \,X_{\text{-}j}$ to perform inference on the conditional
	relationship between $Y$ and $X_j$ given $X_{\text{-}j}$. However, propensity scores require
	$X_j$ be binary, and the propensity score itself is normally
	estimated, although \citet{PR:1984} shows that when all the covariates
	jointly take a small number of discrete values, propensity score
	analysis can be done exactly. \citet{GD--KM-KZ-BS:2014} also rely on
	the data containing repeated observations of $X_{\text{-}j}$ so that
	certain observations can be permuted nonparametrically while
	maintaining the null distribution. In fact, the exact term
	``conditional randomization test'' has also been used in randomized
	controlled experiments to test for independence of $Y$ and $X_j$
	conditioned more generally on some function of $X_{\text{-}j}$ (such
	as a measure of imbalance in $X_{\text{-}j}$ if $X_j$ is binary),
	again relying on discreteness of the function so that there exist
	permutations of $X_j$ which leave the function value
	unchanged. Despite the similar name, our conditional
	randomization test is quite distinct from these, as it does not rely
	on discreteness or experimental control in any of the covariates.
	
	Another line of work exists within the linear model regime, whereby
	the null (without $X_j$) model is estimated and then the empirical
	residuals are permuted to produce a null distribution for $Y$
	\citep{DF-DL:1983}. Because this approach is only exact when the
	empirical residuals match the true residuals, it explicitly relies on
	a parametric model for $Y\,|\, X_{\text{-}j}$, as well as the ability to
	estimate it quite accurately.

	\subsection{Comparisons with knockoffs}
	\label{sec:crtko}
	
	One major limitation of the conditional randomization method is its
	computational cost. It requires computing randomization p-values for
	many covariates and to a high-enough resolution for
	multiple-comparisons correction. Clearly, this requires samples in
	the extreme tail of the p-value distribution.  This means computing a
	very large number of feature importance statistics $T_j$, each of
	which can be expensive since for reasons outlined in the drawbacks
	associated with marginal testing, powerful $T_j$'s will take into
	account the full dimensionality of the model, e.g., absolute value of
	the Lasso-estimated coefficient.  In fact, the number of computations
	of $T_j$, tallied over all $j$, required by the conditional
	randomization method is $\Omega(p)$.\footnote{$a(N) \in \Omega(b(N))$
		means that there exist $N_0$ and $C>0$ such that $a(N) \ge Cb(N)$ for
		$N\ge N_0$.}  To see this, suppose for simplicity that all $R$
	rejected p-values take on the value of half the BHq cutoff equal to
	$\tau = qR/p$, and all we need to do is upper-bound them below
	$\tau$. This means there are $R$ p-values $P_j$ for which plugging
	$K = \infty$ into Algorithm \ref{alg:CRT} would yield $P_j = \tau/2$.
	After $K<\infty$ samples, the approximate p-value (ignoring the {$+1$}
	correction) is distributed as $K^{-1} \, \operatorname{Bin}(K,P_j)$.
	We could then use this binomial count to construct a confidence
	interval for $P_j$. A simple calculation shows that to be reasonably
	confident that $P_j \le \tau$, $K$ must be on the order of
	at least $1/\tau$. Since there are $R$ such p-values, this justifies
	the claim.

	{Note that for knockoffs, the analogous computation of $T$ need only
		be done exactly once. If, for instance, each $T_j$ requires a Lasso
		computation, then the conditional randomization test's computational
		burden} is very challenging for medium-scale $p$ in the thousands
	and prohibitive for large-scale (e.g., genetics) $p$ in the hundreds of
	thousands or millions. We will see in Section~\ref{sec:condrand} that
	there are power gains, along with huge computational costs, to be had
	by using conditional randomization in place of knockoffs, and
	Section~\ref{sec:realdata} will show that the {M\revv{X}} knockoff
	procedure easily scales to large data sets.
	
	Another advantage of {M\revv{X}} knockoffs is its guaranteed control of
	the FDR, whereas the BHq procedure does not offer strict control when
	applied to arbitrarily dependent p-values.

	\section{Numerical simulations}
	\label{sec:simulations}
	In this section we demonstrate the importance, utility, and
        practicality of M\revv{X} knockoffs for high-dimensional
        nonparametric conditional modeling.
	\subsection{Logistic regression p-values}
	\label{sec:logreg}
	Asymptotic maximum likelihood theory promises valid p-values for each
	coefficient in a GLM only when $n \gg p$. However, these approximate
	p-values can usually be computed as long as $n > p$, so a natural
	question arising from high-dimensional applications is whether such
	asymptotic p-values are valid when $n$ and $p$ are both large with
	$p/n \ge 0.1$, for example. We simulated $10^4$ independent design
	matrices ($n=500$, $p=200$) and binary responses from a logistic
	regression for the following two settings:
	\begin{itemize}
		\item[(1)] $(X_1,\dots,X_p)$ is an AR(1) time series with AR
		coefficient $0.5$ and \[Y\, |\, X_1,\dots,X_p \sim \text{Bernoulli}(0.5)\]
		\item[(2)] $(X_1,\dots,X_p)$ is an AR(1) time series with AR
		coefficient $0.5$ and \[Y\, |\, X_1,\dots,X_p \sim \text{Bernoulli}\left(\text{logit}\left(0.08(X_2+\cdots +X_{21})\right)\right)\]
	\end{itemize}
	\begin{figure}\centering
		\subfigure{\label{null_ar1}
			\includegraphics[width=0.45\textwidth]{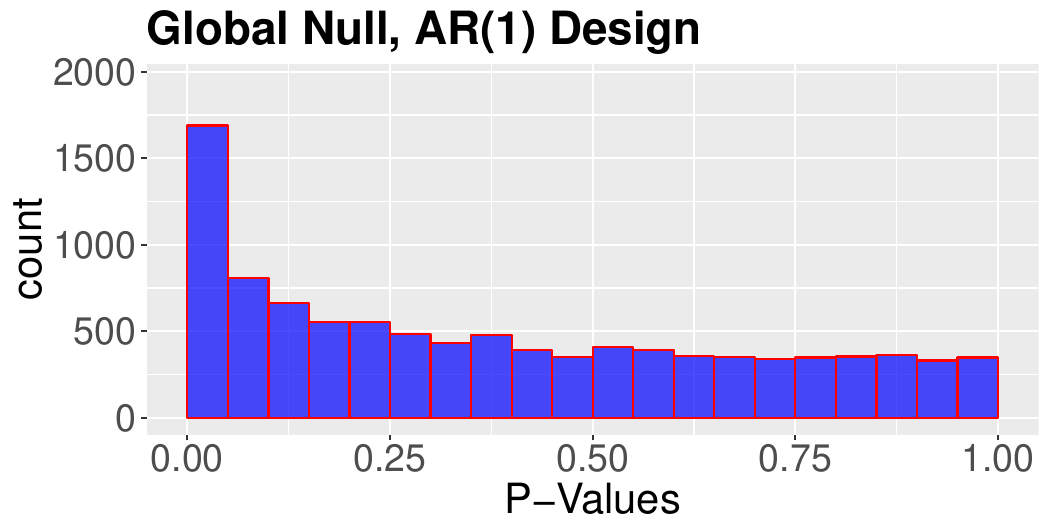}
		}
		\;
		\subfigure{\label{20nz_ar1}
			\includegraphics[width=0.45\textwidth]{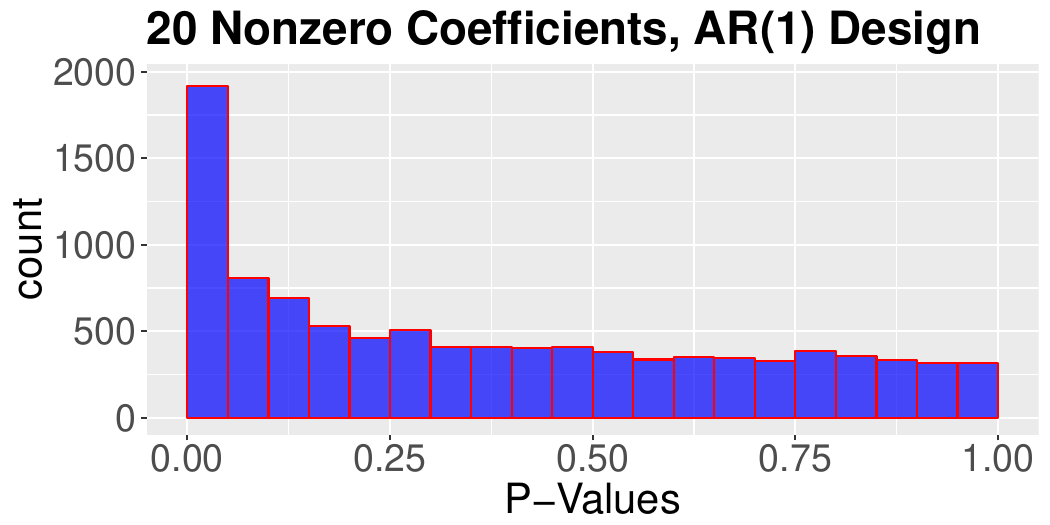}
		}
		\caption{Distribution of null logistic regression p-values with
			$n=500$ and $p=200$; 10,000 replications.}
		\label{LRpvals}
	\end{figure}
	Histograms for the p-values for $\beta_1$ (null in all cases) are
	shown in Figure~\ref{LRpvals}. Both histograms are far from uniform,
	and Table~\ref{LRpvalinfl} shows each distribution's concentration
	near zero. We see that the small quantiles have extremely inflated
	probabilities---over 20 times nominal for
	$\PP{p\text{-value}\le 0.1\%}$ in setting (2). We also see that the
	exact null distribution depends on the unknown coefficient sequence
	$\beta_2,\dots,\beta_p$, since the probabilities between settings
	differ statistically significantly at all three cutoffs.
	
	To confirm that this non-uniformity is not just a finite-sample
	effect, we also simulated $10^4$ i.i.d. $\mathcal{N}(0,1)$ design
	matrices with independent Bernoulli$(0.5)$ responses for $n=500$,
	$p=200$ and $n=5000$, $p=2000$ as settings (3) and (4),
	respectively. Table~\ref{LRpvalinfl} shows that the distribution does
	not really change as $n$ and $p$ are increased with constant
	proportion.  \setlength\extrarowheight{2pt}
	\begin{table}[ht]\centering
\begin{tabular}{|c|c|c|c|c|c|}
\hline
 & (1) & (2) & (3) & (4) \\
\hline
$\PP{p\text{-value}\le 5\%}$ & $16.89\%$ $(0.37\%)$ & $19.17\%$ $(0.39\%)$ & $16.88\%$ $(0.37\%)$ & $16.78\%$ $(0.37\%)$ \\
\hline
$\PP{p\text{-value}\le 1\%}$ & $6.78\%$ $(0.25\%)$ & $8.49\%$ $(0.28\%)$ & $7.02\%$ $(0.26\%)$ & $7.03\%$ $(0.26\%)$ \\
\hline
$\PP{p\text{-value}\le 0.1\%}$ & $1.53\%$ $(0.12\%)$ & $2.27\%$ $(0.15\%)$ & $1.87\%$ $(0.14\%)$ & $2.04\%$ $(0.14\%)$ \\
\hline
\end{tabular}
\caption{Inflated p-value probabilities with estimated Monte Carlo
  standard errors in parentheses. See text for meanings of settings
  (1), (2), (3), (4).}
\label{LRpvalinfl}
\end{table}
	
	These results show that the usual logistic regression p-values one
	might use when $n\ge p$ can have null distributions that are quite far
	from uniform, and even if one wanted to correct that distribution, it
	depends in general on unknown problem parameters, further complicating
	matters. When $n<p$ the problem becomes even more challenging, with
	existing methods similarly asymptotic as well as requiring stringent sparsity assumptions
	\citep{SV-ea:2014}. Thus, despite the wealth of research on
	controlling FDR, without a way to obtain valid p-values, even the problem of controlling
	FDR in medium-to-high-dimensional GLMs remains unsolved.
	
	\subsection{Alternative knockoff statistics}
	As mentioned in Section~\ref{sec:W}, the new \revv{model-X}
        knockoffs framework allows for a wider variety of $W$
        statistics to be used than in the \revv{fixed-X}
        framework. Choices of $Z_j$ include well-studied statistical
        measures such as the coefficient estimated in a GLM, but can
        also include much more ad-hoc/heuristic measures such as
        random forest bagging feature importances or sensitivity
        analysis measures such as the Monte-Carlo-estimated total
        sensitivity index. By providing \revv{variable} selection with rigorous
        Type I error control for general models and statistics,
        knockoffs can be used to improve the interpretability of
        complex black-box supervised/machine learning models. There
        are also many available choices for the anti-symmetric
        function $f_j$, such as ${|Z_j|-|\Zk_j|}$,
        ${\sign(|Z_j|-|\Zk_j|)\max\{|Z_j|,|\Zk_j|\}}$, or
        ${\log(|Z_j|)-\log(|\Zk_j|)}$.
	
	\revv{The main point of this subsection is that knockoffs can be used as a wrapper around essentially \emph{any} data-fitting or prediction algorithm, and regardless of the chosen algorithm still provides rigorous error control for variable selection.} We discuss here a few appealing new options for statistics $W$, but defer full exploration of these very extensive possibilities to future work.
	
	\subsubsection{Adaptive knockoff statistics}
	The default statistic suggested in \citet{RB-EC:2015} is the Lasso
	Signed Max (LSM), which corresponds to $Z_j$ being the largest penalty
	parameter at which the $j$th variable enters the model in the Lasso
	regression of $y$ on $[\bX, \bXk]$, and
	$f_j = {\sign(|Z_j|-|\Zk_j|)\max\{|Z_j|,|\Zk_j|\}}$.  In addition to the
	LSM statistic, \citet{RB-EC:2015} suggested alternatives such as the
	difference in absolute values of estimated coefficients  for
	a variable and its knockoff:
	\[W_j = |\hat{b}_j|-|\hat{b}_{j+p}|,\]
	where the $\hat{b}_j,\hat{b}_{j+p}$ are estimated so that $W$ obeys
	the sufficiency property required by the \revv{FX} knockoff procedure,
	e.g.,~by ordinary least squares or the Lasso with a pre-specified tuning
	parameter. The removal of the sufficiency requirement for {M\revv{X}}
	knockoffs allows us to improve this class of statistics by adaptively
	tuning the fitted model. The simplest example is the LCD statistic
	introduced in Section~\ref{sec:W}, which uses cross-validation to choose
	the tuning parameter in the Lasso. Note the LCD statistic can be
	easily extended to any GLM by replacing the first term in
	\eqref{eq:augmentedLasso} by a non-Gaussian negative log-likelihood, such as in
	logistic regression; we will refer to all such statistics generically
	as LCD. The key is that the tuning and cross-validation is done on the
	augmented design matrix $[\bX,\, \bXk]$, so that $W$ still obeys the
	flip-sign property.
	
	More generally, {M\revv{X}} knockoffs allows us to construct statistics that
	are highly adaptive to the data, as long as that adaptivity does not
	distinguish between original and knockoff variables. For instance, we could compute the
	cross-validated error of the ordinary Lasso (still of $y$ on $[\bX,
	\bXk]$) and compare it to that of a random forest, and choose $Z$ to
	be a feature importance measure derived from whichever one has smaller
	error. Since the Lasso works best when the true model is close to
	linear, while random forests work best in non-smooth models, this
	approach gives us high-level adaptivity to the model smoothness, while
	the {M\revv{X}} knockoff framework ensures strict Type I error control.
	
	Returning to the simpler example of adaptivity, we found the LCD
	statistic to be uniformly more powerful than the LSM statistic across
	a wide range of simulations (linear and binomial GLMs, ranging
	covariate dependence, effect size, sparsity, sample size, total number
	of variables), particularly under covariate dependence. We note,
	however, the importance of choosing the penalty parameter that
	minimizes the cross-validated error, as opposed to the default in some
	computational packages of using the ``one standard error'' rule, as the latter
	causes LCD to be underpowered compared to LSM in low-power
	settings. Figure~\ref{fig:kn_v52} shows a simulation with $n=3000$,
	$p=1000$ of a binomial linear model (with statistics computed from
	Lasso logistic regression) that is representative of the power
	difference between the two statistics. In all our simulations, unless otherwise specified,
	{M\revv{X}} knockoffs is always run using the LCD
	statistic. Explicitly, when the response variable is continuous, we
	use the standard Lasso with Gaussian linear model likelihood, and when
	the response is binary, we use Lasso-penalized logistic regression.
	\begin{figure}\centering
		\subfigure{\label{kn_pwr_v65}
			\includegraphics[width=0.45\textwidth]{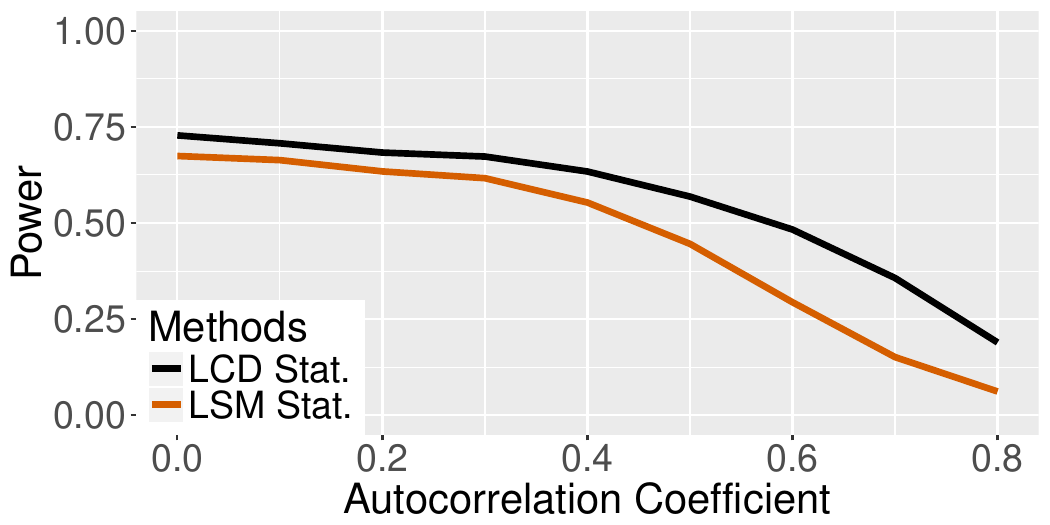}
		}
		\;
		\subfigure{\label{kn_fdr_v65}
			\includegraphics[width=0.45\textwidth]{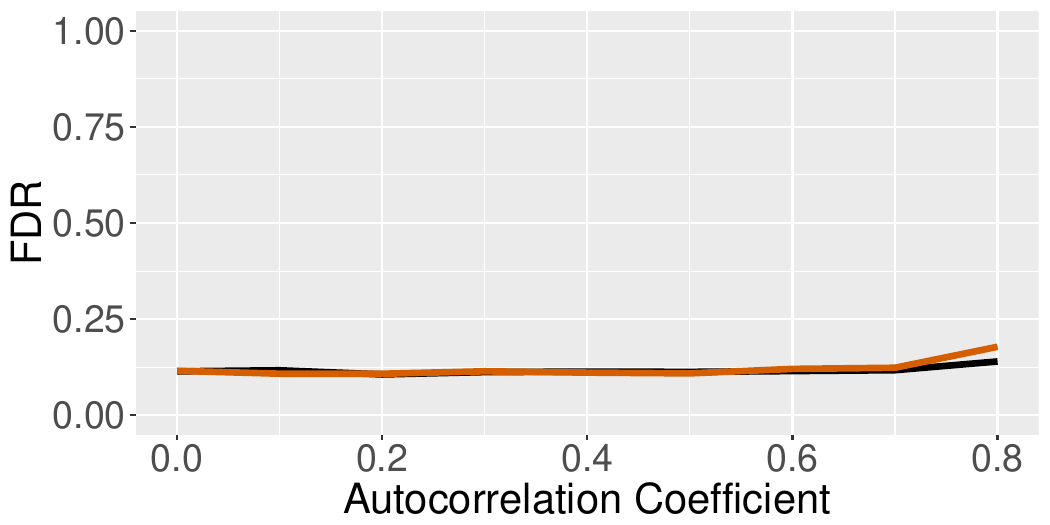}
		}
		\caption{Power and FDR (target is 10\%) for knockoffs with the LCD
			and LSM statistics. The design matrix has i.i.d.~rows and AR(1)
			columns with autocorrelation coefficient specified by the x-axes
			of the plots, and marginally each
			$X_j\sim\mathcal{N}(0,1/n)$. Here, $n=3000$, $p = 1000$, and
			$y$ comes from a {binomial linear model with logit link
				function} with 60 nonzero regression coefficients of magnitude 3.5 and
			random signs. Each point represents
			200 replications.}
		\label{fig:kn_v52}
	\end{figure}
	
	\subsubsection{Bayesian knockoff statistics}
	\label{sec:bvs}
	Another very interesting source of knockoff statistics comes from
	Bayesian procedures. If a statistician has prior knowledge about the
	problem, he or she can encode it in a Bayesian model and use the
	resulting estimators to construct a statistic (e.g.,~difference of
	absolute posterior mean coefficients, or difference or log ratio of
	posterior probabilities of nonzero coefficients with a sparse
	prior). What makes this especially appealing is that the statistician
	gets the power advantages of incorporating prior information, while
	maintaining a strict frequentist guarantee on the Type I error,
	\emph{even if the prior is false}!
	
	As an example, we ran knockoffs in an experiment with a Bayesian
	hierarchical regression model with $n=300$, $p=1000$, and
	$\E{\|\beta\|_0}=60$ ($\|\cdot\|_0$ denotes the $\ell_0$ norm, or the
	number of nonzero entries in a vector); see Appendix~\ref{app:bvs} for
	details. We chose a simple canonical model with Gaussian response to
	demonstrate our point, but the same principle applies to more complex,
	nonlinear, and non-Gaussian Bayesian models as well. The statistics we
	used were the LCD and a Bayesian variable selection (BVS) statistic,
	namely $Z_j - \tilde{Z}_j$ where $Z_j$ and $\tilde{Z}_j$ are the
	posterior probabilities that the $j$th original and knockoff
	coefficients are nonzero, respectively \citep{EG-RM:1997}; again see
	Appendix~\ref{app:bvs} for details. Figure~\ref{fig:bvskn_v4} shows
	that the accurate prior information supplied to the Bayesian knockoff
	statistic gives it improved power over LCD which lacks such
	information, but that they have the same FDR control (and they would
	even if the prior information were incorrect).
	\begin{figure}\centering
		\subfigure{\label{bvskn_pwr_v4}
			\includegraphics[width=0.45\textwidth]{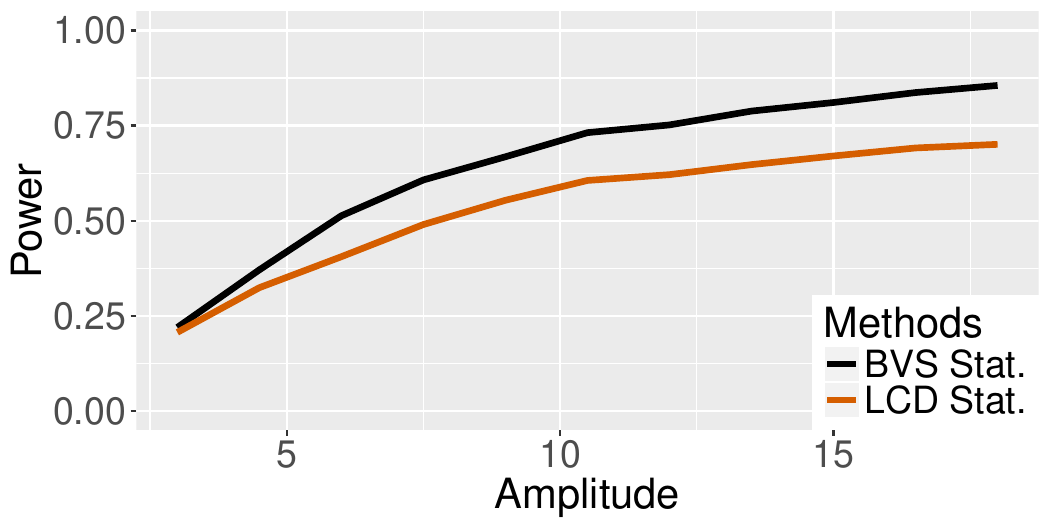}
		}
		\;
		\subfigure{\label{bvskn_fdr_v4}
			\includegraphics[width=0.45\textwidth]{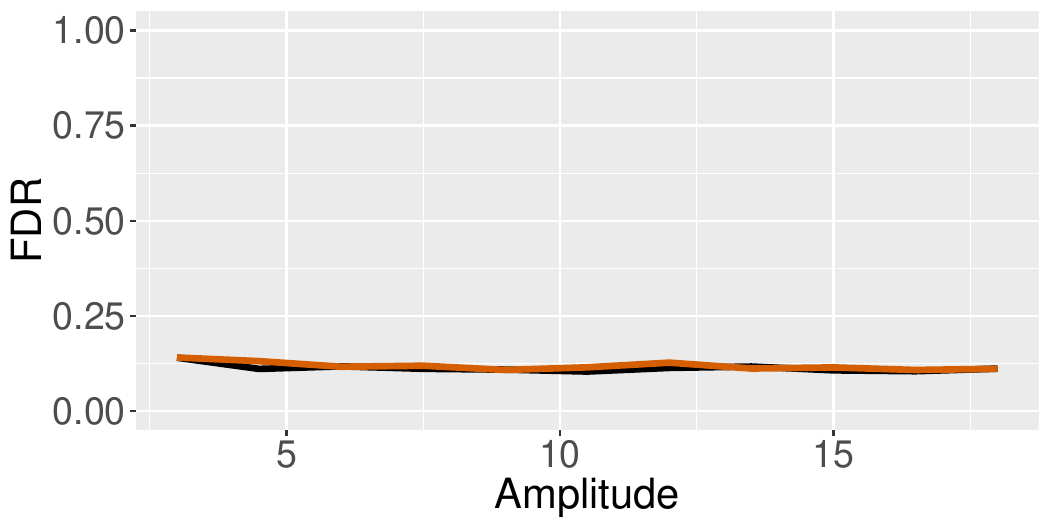}
		}
		\caption{Power and FDR (target is 10\%) for knockoffs with the LCD
			and BVS statistics. The design matrix is i.i.d. $\mathcal{N}(0,1/n)$,
			$n=300$, $p = 1000$, and $y$ comes from a Gaussian linear
			model with $\beta$ and the noise variance randomly chosen (see
			Appendix~\ref{app:bvs} for the precise model). Here, the nonzero
			entries of $\beta$ are Gaussian with mean zero and standard
			deviation given on the x-axis; the expected number of nonzero
			components is 60; the expected variance of the noise is 1.  Each
			point represents 200 replications.}
		\label{fig:bvskn_v4}
	\end{figure}
	
	\subsection{Alternative procedures}
	To assess the relative power of knockoffs, we compare to a number of
	alternatives in settings in which they are valid:
	\begin{itemize}
		\item[1.] The \revv{FX} knockoff procedure with settings recommended in
		\citet{RB-EC:2015}. This method can only be
		applied in homoscedastic Gaussian linear regression when $n\ge p$.
		\item[2.] BHq applied to asymptotic GLM p-values. This method can only
		be applied when $n\ge p$, and although for linear regression exact
		p-values can be computed (when the MLE exists), for any other GLM
		these p-values can be far from valid unless $n\gg p$, as shown in
		\rev{Section~\ref{sec:logreg}}.
		\item[3.] BHq applied to marginal test p-values. The correlation
		between the response and each covariate is computed and compared to its null
		distribution, which under certain Gaussian assumptions is
		closed-form, but in general can at least be simulated exactly by
		conditioning on $y$ and using the known marginal distribution of
		$X_j$. Although these tests are valid for testing hypotheses of \emph{marginal} independence
		(regardless of $n$ and $p$), such hypotheses only agree with the
		desired \emph{conditional} independence hypotheses when the
		covariates are exactly independent of one another.
		\item[4.] BHq applied to the p-values from the conditional
		randomization test described in
                \rev{Section~\ref{sec:crtintro} and Section~\ref{sec:CRT}}.
	\end{itemize}
	\revv{Note that we are using knockoffs, not knockoffs+, in all simulations, and thus we are technically controlling a slightly modified version of FDR. The FDR is nevertheless effectively controlled in all simulations except in extremely low-power settings, and even then the violations are small. We could have sacrificed a small amount of power and used knockoffs+ (both MX and FX) for exact FDR control, but then a more fair comparison in settings 2--4 would replace BHq with the conservative procedure in \citet{YB-DY:2001}, since the joint distribution of the p-values will not in general satisfy the assumptions for BHq to control the FDR exactly. However that conservative procedure had extremely noncompetitive power, so we prefer instead to compare knockoffs and regular BHq, which are more powerful and still effectively control the FDR.}
	
	\subsubsection{Comparison with conditional randomization}
	\label{sec:condrand}
	We start by comparing {M\revv{X}} knockoffs with procedure 4, BHq
	applied to conditional randomization test p-values, for computational
	reasons. We simulated $n=400$ i.i.d. rows of $p=600$ AR(1) covariates
	with autocorrelation 0.3, and response following a logistic regression
	model with 40 nonzero coefficients of random
	signs. Figure~\ref{fig:rand_v9} shows the power and FDR curves as the
	coefficient amplitude was varied.
	\begin{figure}\centering
		\subfigure{\label{rand_pwr_v9}
			\includegraphics[width=0.45\textwidth]{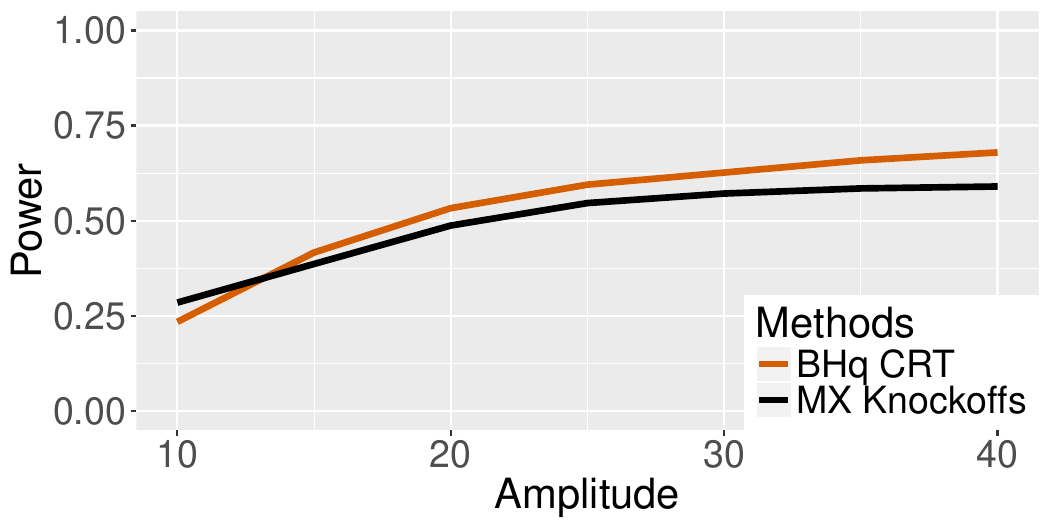}
		}
		\;
		\subfigure{\label{rand_fdr_v9}
			\includegraphics[width=0.45\textwidth]{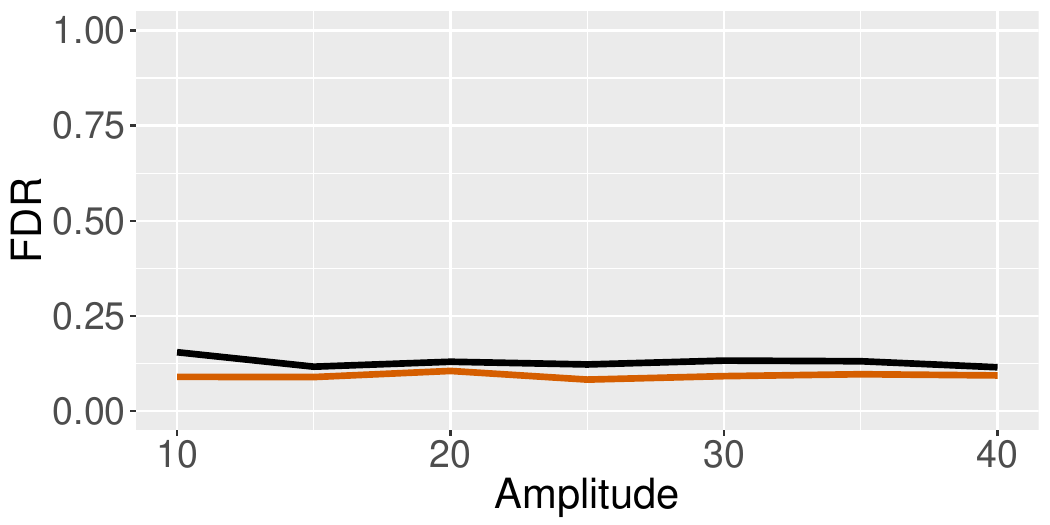}
		}
		\caption{Power and FDR (target is 10\%) for {M\revv{X}} knockoffs
			and BHq applied to conditional randomization test p-values. The
			design matrix has i.i.d. rows and AR(1) columns with
			autocorrelation 0.3, $n=400$, $p = 600$, and $y$ comes from a
			binomial linear model with logit link function with
			$\|\beta\|_0=40$, and
			all nonzero entries of $\beta$ having equal magnitudes and random
			signs; each point represents 200 replications.}
		\label{fig:rand_v9}
	\end{figure}
	We see that the conditional randomization test gives higher
        power with similar FDR control, but this comes at a hugely
        increased computational cost. This simulation has considerably
        smaller $n$ and $p$ than any other simulation in the paper,
        and we still had to apply a number of computational
        speed-ups/shortcuts, described in Appendix~\ref{app:CRspeed}, to keep the computation time within
        reason.
	
	With these speed-ups, Figure~\ref{fig:rand_v9} took roughly three
	years of serial computation time, while the {M\revv{X}} knockoffs
	component took only about six hours, or about 5000 times less (all
	computation was run in Matlab 2015b, and both methods used glmnet to
	compute statistics). Because of the heavy computational burden, we
	were unable to include the conditional randomization test in our
	further, larger simulations---\rev{we show in Section~\ref{sec:crtko}} that
	the number of $T_j$ computations scales optimistically linearly in
	$p$. To summarize, conditional randomization testing appears somewhat
	more powerful than {M\revv{X}} knockoffs, but is computationally
	infeasible for large data sets (like that in
	Section~\ref{sec:realdata}).
	
\subsubsection{Effect of signal amplitude}
Our first simulation comparing {M\revv{X}} knockoffs to procedures 1--3 is by
necessity in a Gaussian linear model with $n>p$ and independent
covariates---the only setting in which all procedures approximately
control the FDR. Specifically, the left side of
Figure~\ref{kn_v44_45} plots the power and FDR for the four
procedures when
$X_{ij} \iid \mathcal{N}(0,1/n)$, $n=3000$, $p=1000$,
$\|\beta\|_0 = 60$, the noise variance $\sigma^2=1$, and the nonzero
entries of $\beta$ have random signs and equal magnitudes, varied
along the x-axis. All methods indeed control the FDR, and {M\revv{X}}
knockoffs is the most powerful, with as much as 10\% higher power than
its \emph{nearest} alternative. The right side of Figure~\ref{kn_v44_45} shows the same setup but in high dimensions:
$p=6000$. In the high-dimensional regime, neither maximum likelihood
p-values nor \revv{FX} knockoffs can even be computed, and {the {M\revv{X}}
  knockoff procedure} has considerably higher power than BHq applied
to marginal p-values.
\begin{figure}\centering
    \subfigure{\label{kn_v44}
      \includegraphics[width=0.45\textwidth]{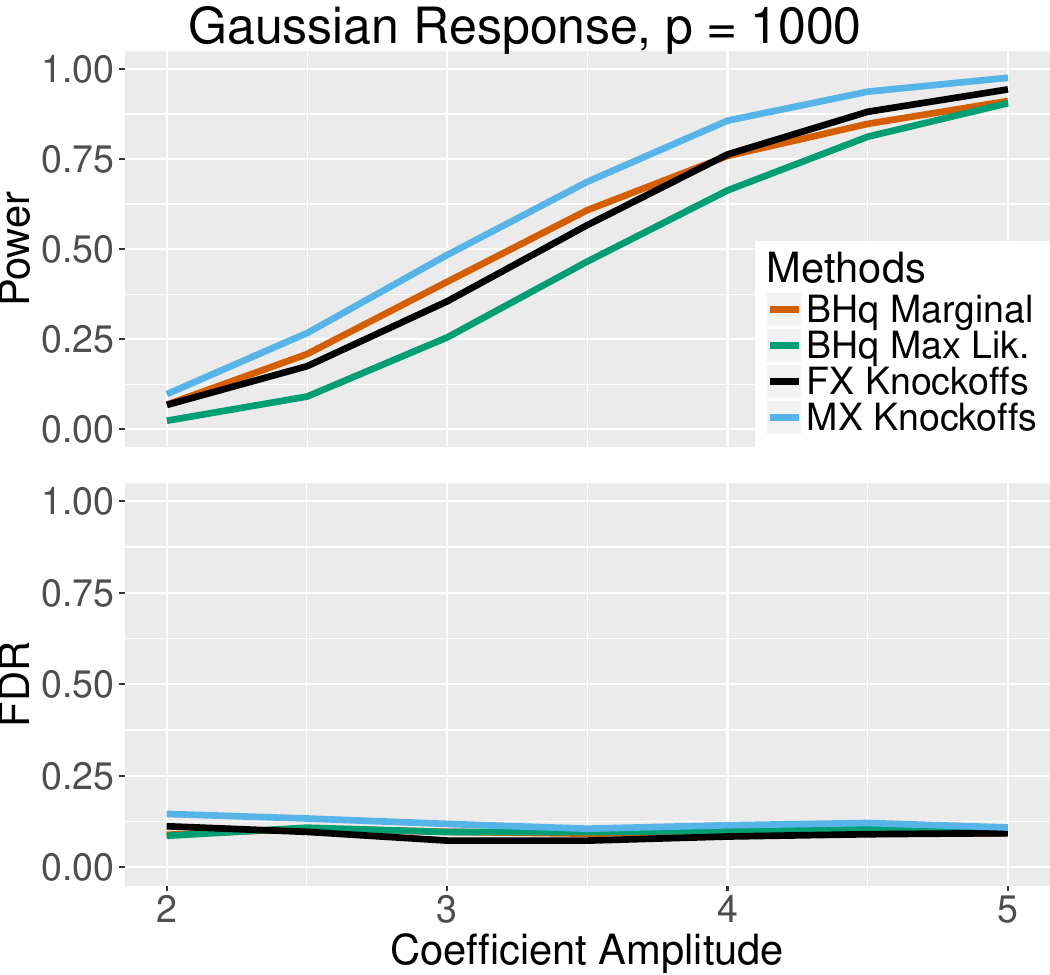}
    }
    \;
    \subfigure{\label{kn_v45}
      \includegraphics[width=0.45\textwidth]{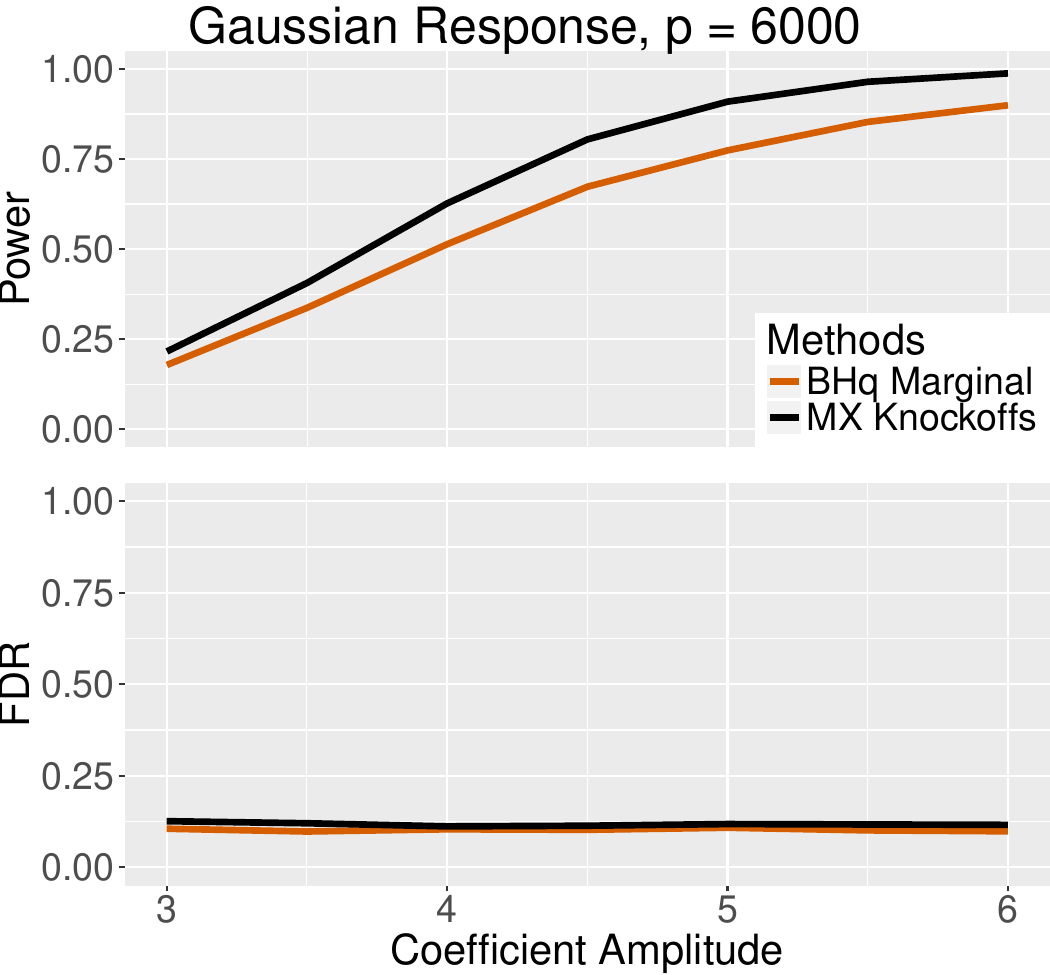}
    }

    \caption{Power and FDR (target is 10\%) for {M\revv{X}} knockoffs
      and alternative procedures. The design matrix is
      i.i.d. $\mathcal{N}(0,1/n)$, $n=3000$, $p = $ (left): $1000$ and (right):
      $6000$, and $y$ comes from a Gaussian linear model with 60
      nonzero regression coefficients having equal magnitudes and
      random signs. The noise variance is 1.  Each point represents
      200 replications.}
\label{kn_v44_45}
\end{figure}

Next we move beyond the Gaussian linear model to a binomial linear
model with logit link function, precluding the use of the original
knockoff procedure. Figure~\ref{kn_v46_47} shows the same simulations
as Figure~\ref{kn_v44_45} but with $Y$ following the binomial
model. The results are similar to those for the Gaussian linear model, except
that BHq applied to the asymptotic maximum likelihood p-values now has
an FDR above 50\% (rendering its high power meaningless), which can be
understood as a manifestation of the phenomenon from \rev{Section~\ref{sec:logreg}}. In
summary, {M\revv{X} knockoffs} continues to have the highest power among
FDR-controlling procedures.
\begin{figure}\centering
    \subfigure{\label{kn_v46}
      \includegraphics[width=0.45\textwidth]{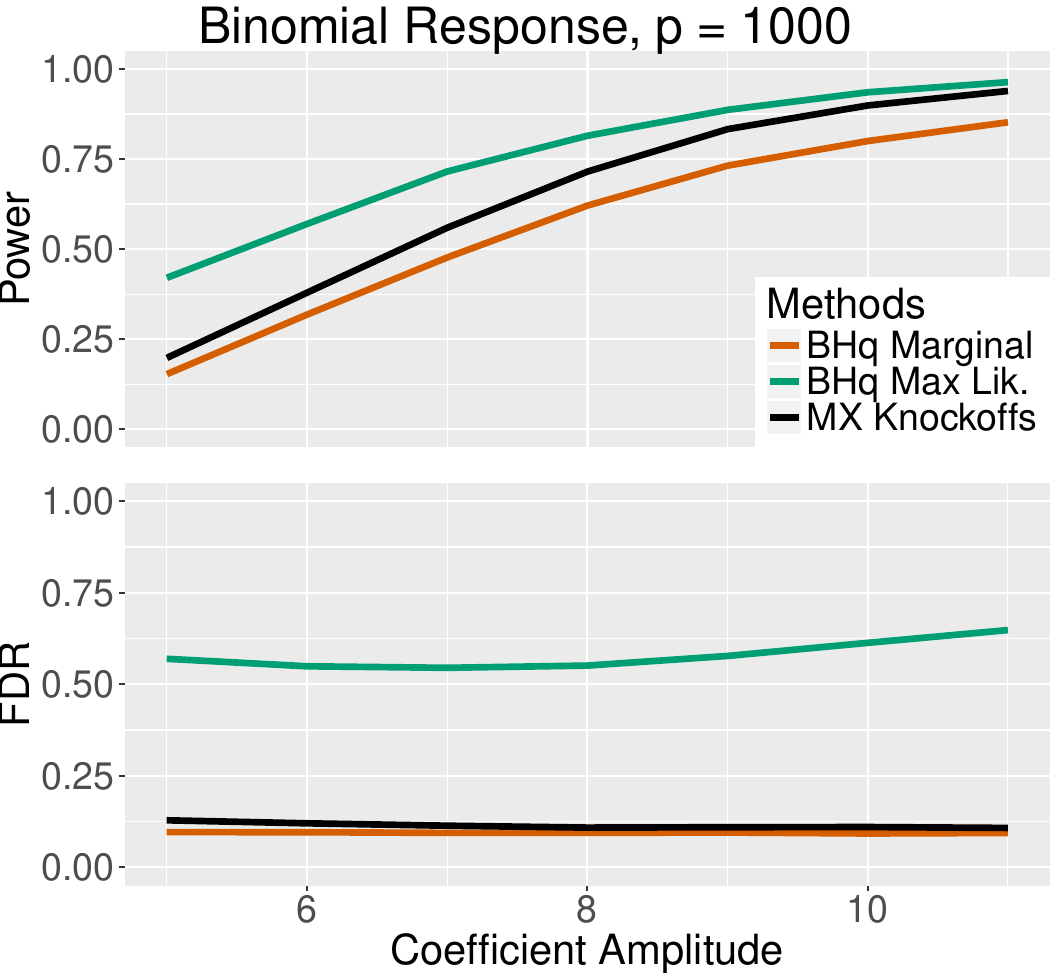}
    }
    \;
    \subfigure{\label{kn_v47}
      \includegraphics[width=0.45\textwidth]{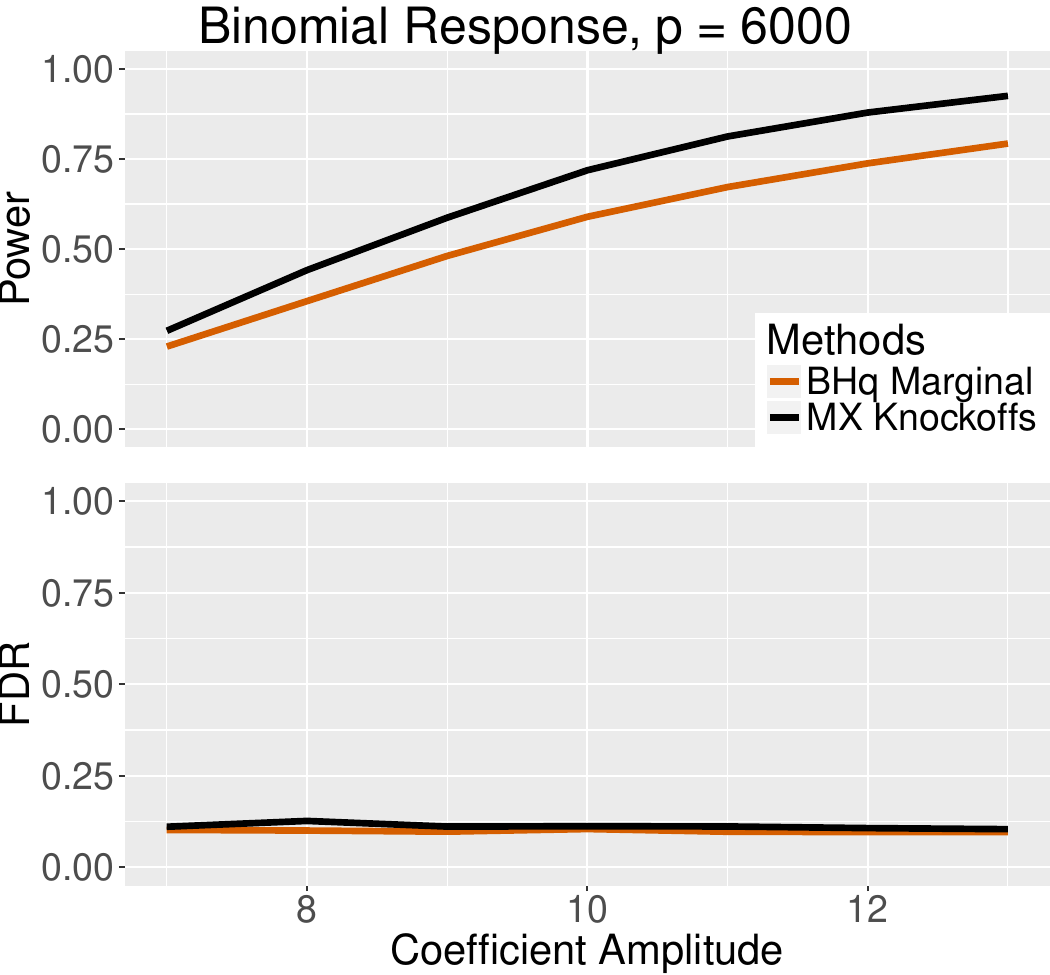}
    }

    \caption{Power and FDR (target is 10\%) for {M\revv{X}} knockoffs
      and alternative procedures. The design matrix is
      i.i.d. $\mathcal{N}(0,1/n)$, $n=3000$, $p = $ (left): $1000$ and (right):
      $6000$, and $y$ comes from a binomial linear model with logit
      link function, and 60 nonzero regression coefficients having
      equal magnitudes and random signs. Each point represents 200
      replications.}
\label{kn_v46_47}
\end{figure}

\subsubsection{Effect of covariate dependence}
To assess the relative power and FDR control of {M\revv{X}} knockoffs
as a function of covariate dependence, we ran similar simulations as
in the previous section, but with covariates that are AR(1) with varying
autocorrelation coefficient (while the coefficient amplitude remains
fixed). It is now relevant to specify that the locations of the
nonzero coefficients are uniformly distributed on $\{1,\dots,p\}$. In
the interest of space, we only show the low-dimensional ($p=1000$)
Gaussian setting (where all four procedures can be computed) and the
high-dimensional ($p=6000$) binomial setting, as little new
information is contained in the plots for the remaining two
settings. Figure~\ref{kn_v51_49} shows that, as
expected, BHq with marginal testing quickly loses FDR control with
increasing covariate dependence. This is because the marginal tests
are testing the null hypothesis of \emph{marginal} independence
between covariate and response, while recall from
Definition~\ref{def:null} that all conditionally independent
covariates are considered null, even if they are marginally dependent
on the response. Concentrating on the remaining methods and just the
left-hand part of the BHq marginal curves where FDR is controlled,
Figure~\ref{kn_v51_49} shows that {M\revv{X} knockoffs} continues to be considerably more
powerful than alternatives as covariate dependence is introduced, in
low- and high-dimensional linear and nonlinear models.
\begin{figure}\centering
    \subfigure{\label{kn_v51}
      \includegraphics[width=0.45\textwidth]{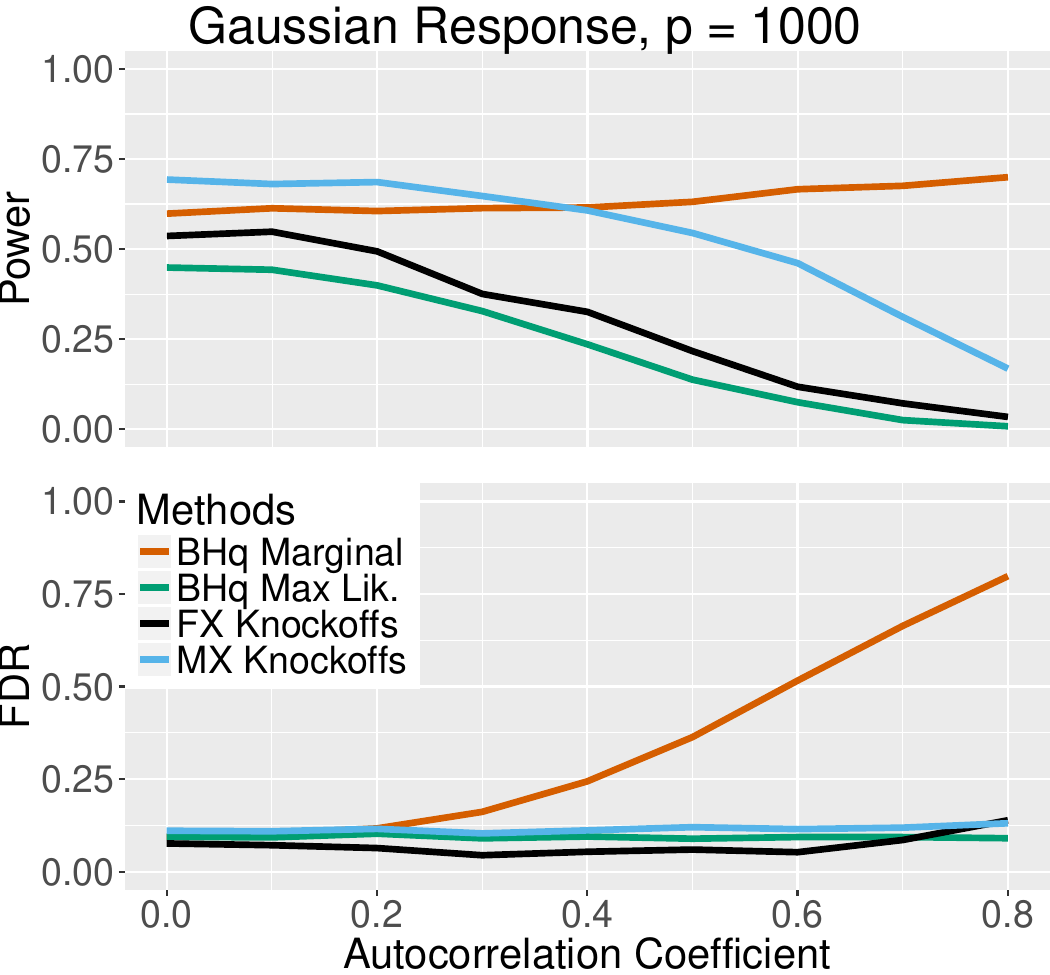}
    }
    \;
    \subfigure{\label{kn_v49}
      \includegraphics[width=0.45\textwidth]{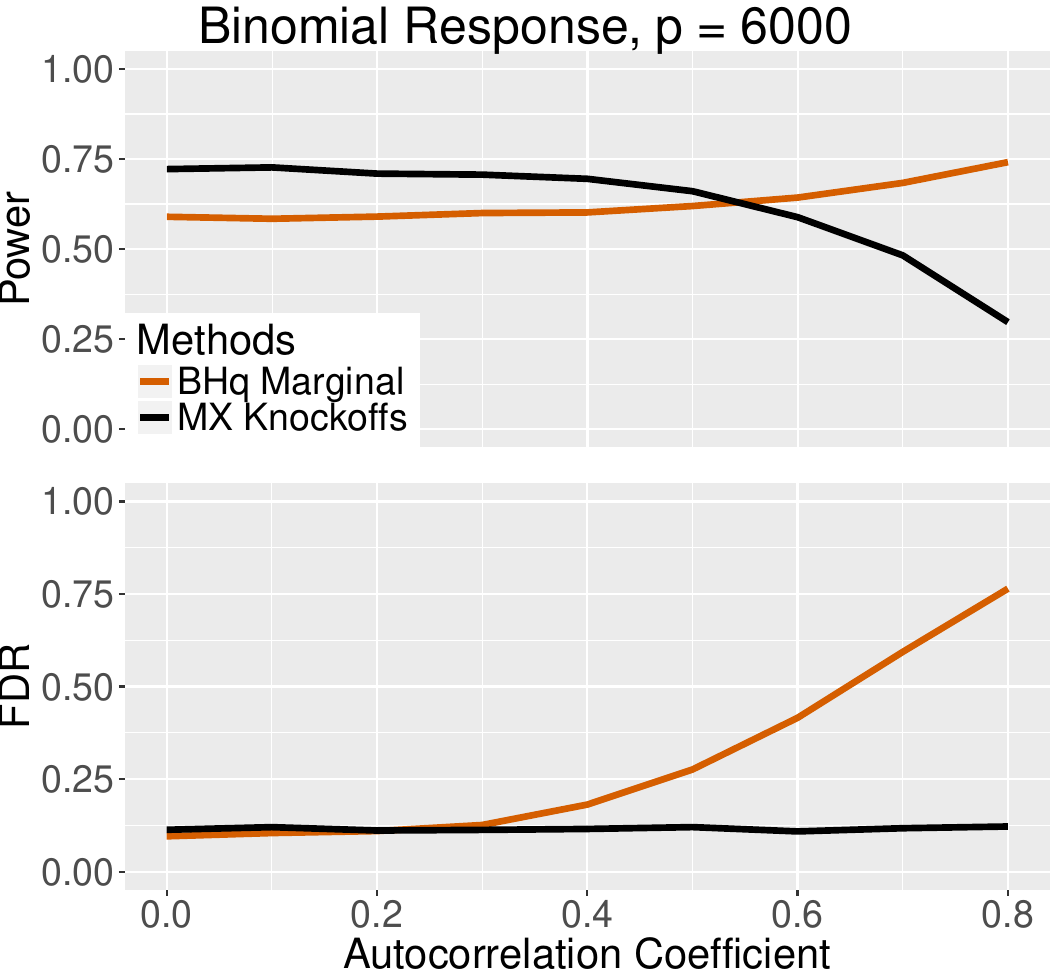}
    }

    \caption{Power and FDR (target is 10\%) for {M\revv{X}} knockoffs
      and alternative procedures. The design matrix has i.i.d.~rows
      and AR(1) columns with autocorrelation coefficient specified by
      the x-axes of the plots, and marginally each
      $X_j\sim\mathcal{N}(0,1/n)$. (left): $n = 3000$, $p=1000$, and
      $y$ follows a Gaussian linear model.  (right): $n = 3000$,
      $p=6000$, and $y$ follows a binomial linear model with logit
      link function. In both cases, there are 60 nonzero coefficients
      having magnitudes equal to 3.5 on the left and 10 on the right,
      random signs, and randomly selected locations. Each point
      represents 200 replications.}
\label{kn_v51_49}
\end{figure}
	
        \rev{
	\section{Robustness}
	\label{sec:robust}
	In many real applications, the true joint covariate distribution may
	not be known exactly, forcing the user to estimate it from the
	available data. As already mentioned, this is a challenging problem by
	itself, but often we have considerable outside information or
	unsupervised data that can be brought to bear to improve
	estimation. This raises the important question of how robust
	{M\revv{X} knockoffs} is to error in the joint covariate
        distribution. Theoretical guarantees of robustness are
          beyond the scope of this paper, but we present instead three
          compelling simulation studies to demonstrate robustness. The
          first study investigates error that biases that distribution toward the empirical
	covariate distribution, often referred to as overfitting
        error, on simulated data. We generated knockoffs for Gaussian variables,
	but instead of using the true covariance matrix, we used
        in-sample covariance estimates which ranged in
          overfitting error. Figure~\ref{fig:kn_v59} shows the power and FDR as the
	covariance we use ranges from the true covariance matrix (AR(1) with
	autocorrelation 0.3), to a graphical lasso estimator, to
        convex combinations of the true and empirical covariance (see
        Appendix~\ref{app:robust} for explicit formulae for the estimators). The
        plot is indexed on the x-axis by the average relative
        Frobenius norm $\|\hat{\bS}-\bS\|_{\text{Fro}}/\|\bS\|_{\text{Fro}}$
        of the estimator $\hat{\bS}$.
	\begin{figure}\centering
		\subfigure{\label{kn_pwr_robust}
			\includegraphics[width=0.45\textwidth]{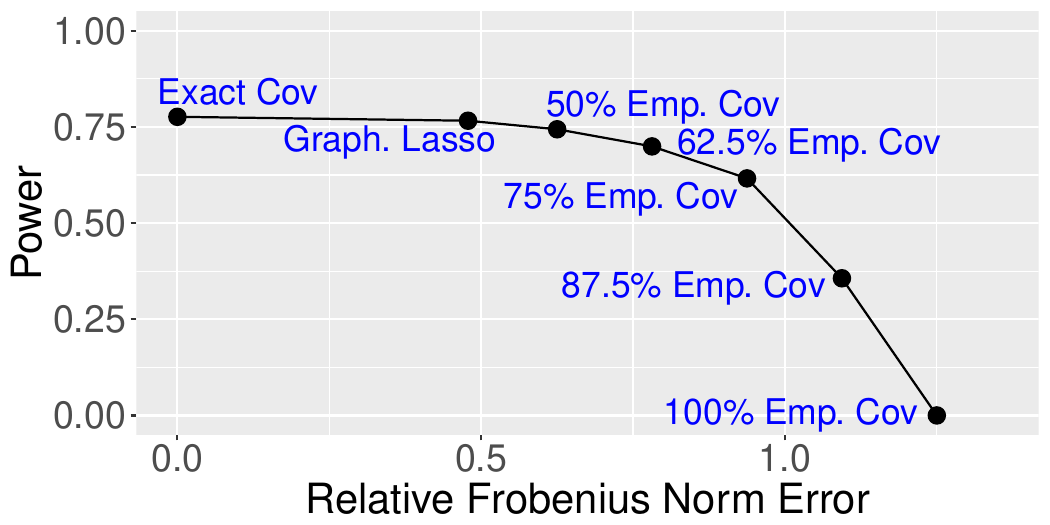}
		}
		\;
		\subfigure{\label{kn_fdr_robust}
			\includegraphics[width=0.45\textwidth]{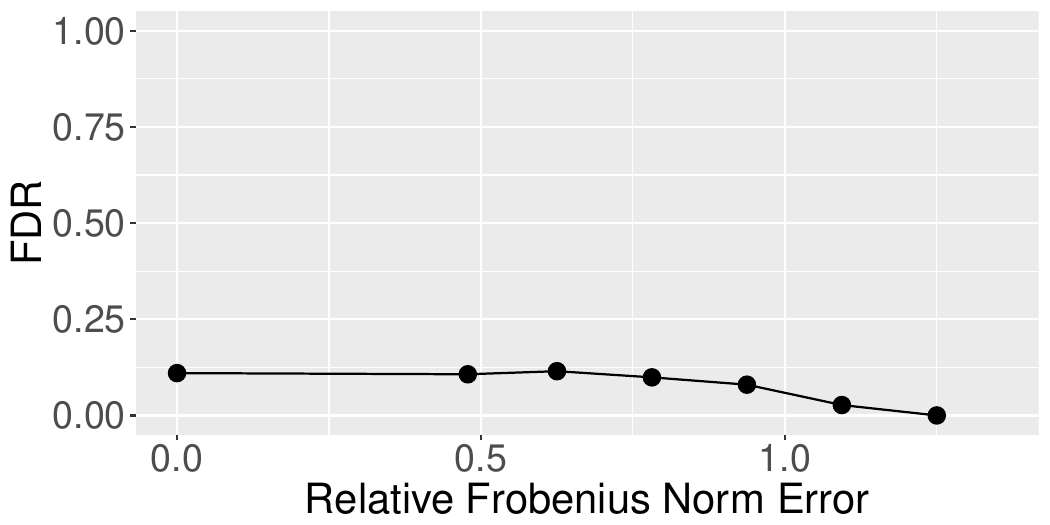}
		}
		\caption{\rev{Power and FDR (target is 10\%) for knockoffs with the LCD
			statistic as the covariance matrix used to generate knockoffs
			ranges from the truth to an estimated covariance; see text for
			details. The design matrix has i.i.d. rows and AR(1) columns
			with autocorrelation coefficient 0.3, and the matrix (including
			knockoffs) is standardized so that each column has mean zero and
			Euclidean norm 1. Here, $n=800$, $p = 1500$, and $y$ comes
			from a binomial linear model with logit link function with 50
			nonzero entries having magnitude 20 and random signs. Each point
			represents 200 replications.}}
		\label{fig:kn_v59}
	\end{figure}
	Although the graphical Lasso is well suited for this problem
        since the covariates have a sparse precision matrix, its
        covariance estimate is still off by nearly 50\%, and yet
        surprisingly the resulting power and FDR are nearly
        indistinguishable from when the exact covariance is used. The
        covariance estimate worsens as the empirical covariance---a
        very poor estimate of the true covariance given the high
        dimensionality---is combined in increasing proportion with the
        truth. At 75\% weight on the empirical covariance, the
        covariance estimate is nearly 100\% off and yet the power and FDR of M\revv{X}
        knockoffs are only slightly decreased. Beyond this point, M\revv{X}
        knockoffs becomes quite conservative, with power and FDR
        approaching zero as the estimated covariance approaches the
        empirical covariance. This behavior at 100\% weight on the
        empirical covariance is not
        surprising, since $p>n$ and thus the empirical covariance is
        rank-deficient, forcing the knockoff variables to be exact
        replicas of their original counterparts.\footnote{\revv{When the knockoff variables are exact copies of the original variables we are guaranteed zero power and zero FDR since all $W_j=0$.} \rev{Although in
          principle we could break ties and assign signs by coin flips
          when $W_j=0$, we prefer to only select $X_j$ with $W_j>0$,
          as $W_j=0$ provides no evidence against the null
          hypothesis.}} The main take-aways from this plot are that (a)
        the nominal level of 10\% FDR is never violated, even for
        covariance estimates very far from the truth, and (b) the
        more overfitting done on the covariance, the more conservative
        the procedure is, although even at almost 100\% relative
        error, M\revv{X} knockoffs had only lost about 20\% of the power it
        would have had if the covariance were known exactly. \revv{Intuitively, instead of treating the variables as coming from their true joint distribution, M\revv{X} knockoffs with an overfit covariate distribution seems to treat them as coming from their true distribution ``conditional" on being similar to their observed values. Thus FDR should be roughly controlled \emph{conditionally} which implies marginal FDR control, while power may be lost if the conditioning is too great, which matches what we see in the simulations.}

        Our second and third experiments use real covariate data from
        a genome-wide association study (GWAS), the details of which
        are given in Section~\ref{sec:realdata} and
        Appendix~\ref{app:realdata}. In brief, it is a
        high-dimensional setting with $X_{ij}\in\{0,1,2\}$ and strong
        spatial structure, whose covariance we estimate in-sample
        using the GWAS-tailored covariance estimator of
        \citet{XW-MS:2010}. We check the robustness of constructing
        second-order model-\revv{X} knockoffs by ASDP (from
        Section~\ref{sec:asdp}) by choosing a reasonable but
        artificial model for $Y|X_1,\dots,X_p$ and simulating
        artificial response data using the real covariate data. The
        exact details of the simulation are given in
        Section~\ref{app:simrealdata}, but note that this simulation
        used the exact same covariance estimation, SNP clustering and
        representative selection, knockoff construction, and knockoff
        selection procedure as used for the real data analysis of the
        next section. Our second experiment varies the signal
        amplitude in a binomial linear model, and
        Figure~\ref{fig:wtccc_fake} shows the FDR and power. As hoped,
        the FDR is consistently controlled over a wide range of
        powers.
	\begin{figure}\centering
		\subfigure{\label{wtccc_pwr_fake_73_78}
			\includegraphics[width=0.45\textwidth]{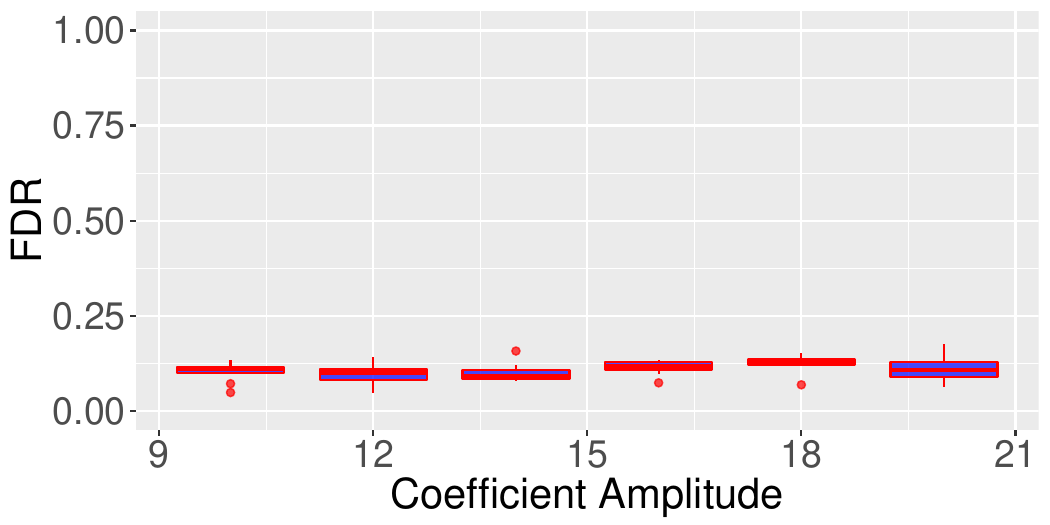}
		}
		\;
		\subfigure{\label{wtccc_fdr_fake_73_78}
			\includegraphics[width=0.45\textwidth]{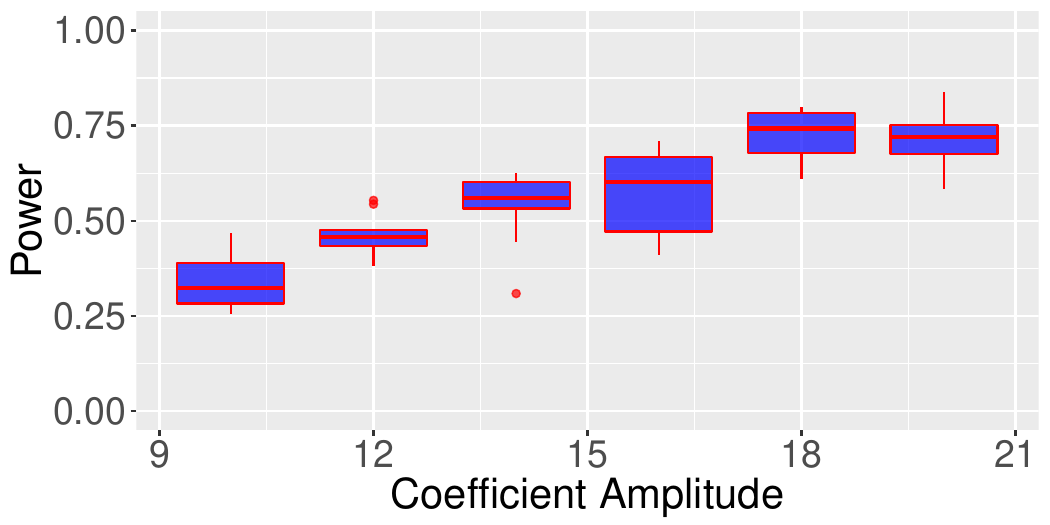}
		}
		\caption{\rev{Power and FDR (target is 10\%) for knockoffs
                  with the LCD statistic applied to subsamples of a
                  real genetic design matrix. Each boxplot represents
                  10 different logistic regression models with 60
                  nonzero coefficients \rev{with amplitudes given by the
                    x-axis}, and for each model, 1000
                  common observations were used for picking cluster
                  representatives, and the remaining 13,708
                  observations were divided into 10 disjoint parts and
                  knockoffs run on each part, with power and FDR for
                  that model then computed by averaging the results
                  over those 10 parts.}}
		\label{fig:wtccc_fake}
	\end{figure}
        Our third experiment, instead of varying the signal strength
        of our artificial model for $Y|X_1,\dots,X_p$, deliberately
        corrupts the covariance estimate of \citet{XW-MS:2010} by
        varying a shrinkage parameter that is not meant to be
        varied. That parameter is $m$ in \citet[equation
        (2.7)]{XW-MS:2010}, and we vary it from $1/10$ to $10$ times
        its intended value. This variation has a huge effect on how
        much shrinkage is applied off the diagonal, with the average
        correlation varying by a factor of about 13 over the range of
        shrinkage. Figure~\ref{fig:wtccc_fake_tuning} shows that, even
        as we range from substantial undershrinkage to substantial
        overshrinkage, M\revv{X} knockoffs never significantly violates FDR
        control, with only a bit of conservativeness when the
        undershrinkage is most drastic (the same phenomenon as the
        right side of Figure~\ref{fig:kn_v59}).
	\begin{figure}\centering
		\subfigure{\label{wtccc_pwr_fake_79_84}
			\includegraphics[width=0.45\textwidth]{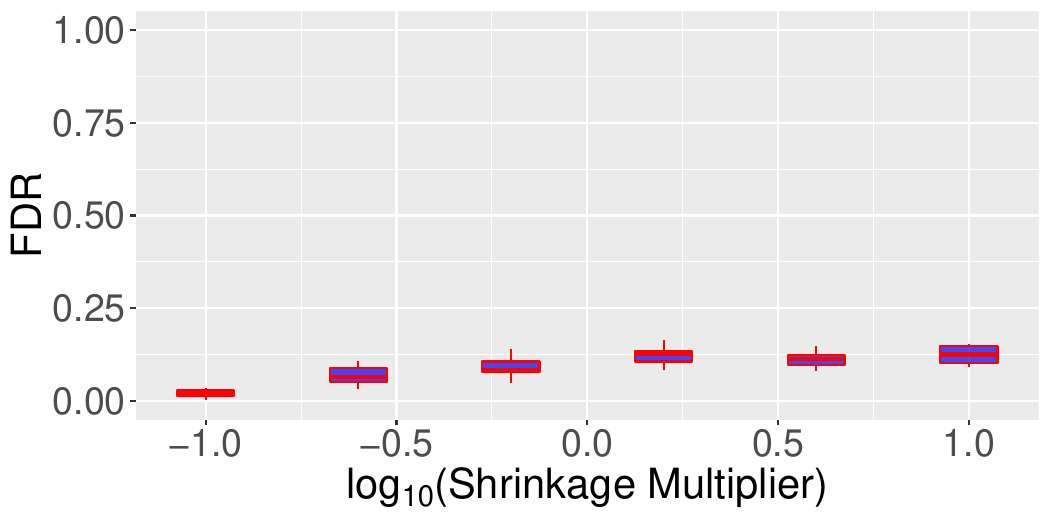}
		}
		\;
		\subfigure{\label{wtccc_fdr_fake_79_84}
			\includegraphics[width=0.45\textwidth]{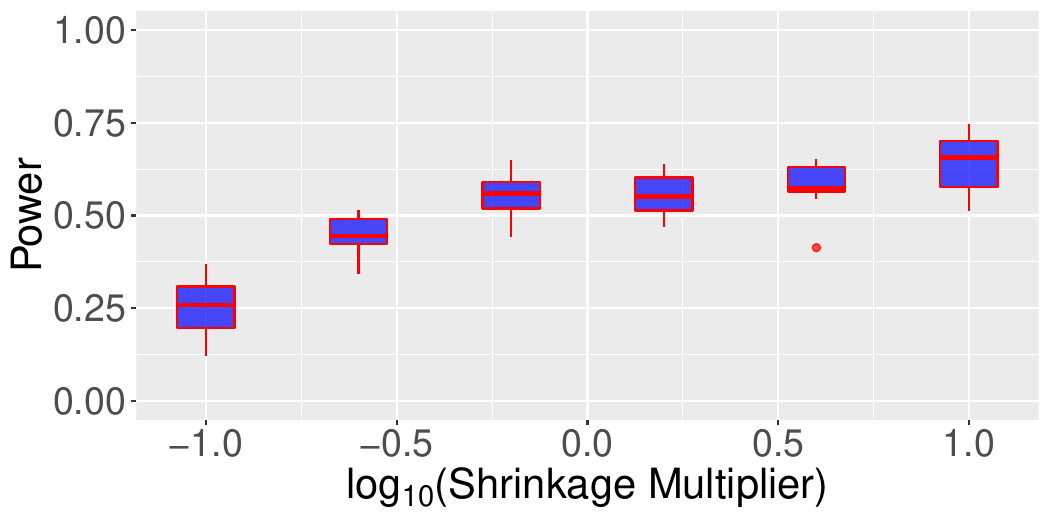}
		}
		\caption{\rev{The setup is the same as
                  Figure~\ref{fig:wtccc_fake} except with the
                  amplitude fixed at 14 and the amount of off-diagonal
                shrinkage in the covariance estimate varied on the x-axis.}}
		\label{fig:wtccc_fake_tuning}
	\end{figure}

        }

	\section{Genetic analysis of Crohn's disease}
	\label{sec:realdata}
	To test the robustness and practicality of the new knockoff
	procedure, we applied it to a data set containing genetic information
	on cases and controls for Crohn's disease (CD). The data is provided
	by the Wellcomme Trust Case Control Consortium and has been studied
	previously \citep{WTCCC:2007}. 	After preprocessing, there were $p=377,749$ single nucleotide
	polymorphisms (SNPs) measured on $n=4,913$ subjects (1,917 CD patients
	and 2,996 healthy controls). Although $p \gg n$, the joint dependence
	of SNPs has a strong spatial structure, and outside data can be used
	to improve estimation. In particular, we approximated the standardized
	joint distribution as multivariate Gaussian with covariance matrix
	estimated using the methodology of \citet{XW-MS:2010}, which shrinks
	the off-diagonal entries of the empirical covariance matrix using
	genetic distance information estimated from the HapMap CEU
	population. This approximation was used on each chromosome, and SNPs
	on different chromosomes were assumed to be independent. The statistic
	we use is the LCD. Although the data itself cannot be made available, all
	code is available at \revv{\url{http://web.stanford.edu/group/candes/knockoffs/software/knockoff/}}.
	
	One aspect of SNP data is that it contains some very high
	correlations, which presents two challenges to our methodology. The
	first is generic to the variable selection problem: it is very hard to
	choose between two or more nearly-identical (highly-correlated)
	variables if the data supports at least one of them being
	selected.\footnote{This is purely a problem of power and would not
		affect the Type I error control of knockoffs.} To alleviate this, we
	clustered the SNPs using the estimated correlations as a similarity
	measure with a single-linkage cutoff of 0.5, and settle for
	discovering important SNP clusters. To do so we choose one
	representative from each cluster and approximate the null hypothesis
	that a cluster is conditionally independent of the response given the
	other clusters by the null hypothesis that a cluster
	\emph{representative} is conditionally independent of the response
	given the other cluster \emph{representatives}. To choose the
	representatives, we could ignore the observed data altogether and do
	something like pick representatives with the highest minor allele
	frequency (computed from outside data), and then run knockoffs as
	described in the paper. Although this produces a powerful procedure
	(about 50\% more powerful than the original analysis by
	\citet{WTCCC:2007}), a more powerful approach is to select cluster
	representatives using a fraction of the observations, including their
	responses, such as by marginal testing. Note that such a
	data-splitting approach appears to make our null hypotheses random, as
	in the work on inference after selection reviewed in
	Section~\ref{sec:related}. However, the approximation we are making is
	that each representative stands for its cluster, and each cluster has
	exactly one associated null hypothesis, no matter how selection is
	performed, even if it were nonrandom. That is, the \emph{approximate}
	hypotheses being tested do not actually depend on the selection
	(unlike \citet{DB-ea:2016} where clusters are selected and where the
	very definition of a cluster actually depends on the selection), and
	our approach remains \revv{free of a model for $Y\mid X$}, which together should make it clear
	that it is still quite different from the literature on inference
	after selection.
	
	Explicitly, we randomly chose 20\% of our observations and on those
	observations only, we ran marginal t-tests between each SNP and the
	response. Then from each cluster we chose the SNP with smallest t-test
	p-value to be the single representative of that cluster. Because the
	observations used for selecting cluster representatives have had their
	representative covariate values selected for dependence on the
	outcome, if we constructed knockoff variables as usual and included
	them in our procedure, the required exchangeability established in
	Lemma~\ref{lem:exch} would be violated. However, taking a page from
	\citet{RB-EC:2016}, we can still use these observations in our
	procedure by making their knockoff variables just identical copies of
	the original variables (just for these observations). It is easy to
	show (see Appendix~\ref{app:knselect}) that constructing knockoffs in
	this way, as exact copies for the observations used to pick
	representatives and as usual for the remaining observations, the
	pairwise exchangeability of null covariates with their knockoffs is
	maintained. Of course, the observations with identical original and
	knockoff covariate values do not directly help the procedure
	distinguish between the original and knockoff variables, but including
	them improves the accuracy of the fitting procedure used to produce
	the feature importance statistics, so that power is improved
	indirectly because the $Z_j$ become more accurate measures of feature
	importance.
	
	Replacing clusters with single representatives
	reduces $p$ down to $71,145$ (so the average cluster size was just
	over five SNPs, although there was substantial variance) and, by
	construction, upper-bounds pairwise SNP correlations by 0.5. Note that
	this is far from removing all dependence among the SNPs, so
	considering conditional independence instead of marginal dependence
	remains necessary for interpretation.  Scientifically, we consider a
	selected SNP to be a true discovery if and only if it is the
	representative of a cluster containing a truly important SNP.
	
	The second challenge is the one discussed in
	Section~\ref{sec:construction}, and we use the approximate SDP
	knockoff construction proposed there. The approximate covariance
	matrix was just the estimated covariance matrix with zeros outside of
	the block diagonal, with the blocks chosen by single-linkage
	clustering on the estimated correlation matrix, aggregating clusters
	up the dendrogram as much as possible subject to a maximum cluster
	size of 999. In this case, even separating the problem by chromosome,
	the SDP construction was computationally infeasible and the
	equicorrelated construction produced extremely small $s$:
	$\text{mean}(s^{\text{EQ}}) = 0.08$.
	The parallelized approximate SDP construction took just a matter of
	hours to run, and increased $s$ on average by almost an order of magnitude,
	with $\text{mean}(s^{\text{ASDP}}) = 0.57$.
	
	Although it incorporates strong prior information, our estimate of the
	joint distribution of the SNPs is still an approximation, so our
	first step is to test the robustness of knockoffs to this
	approximation.
	
	\subsection{Simulations with genetic design matrix}
	\label{app:simrealdata}
	\rev{In Section~\ref{sec:robust}, the second and third
          experiments} check the robustness of knockoffs to this
        particular joint distribution approximation \rev{by taking} a
        reasonable model for $Y|X_1,\dots,X_p$ and simulat\rev{ing}
        artificial response data using the real covariate data
        itself. We split the rows of our design matrix into 10 smaller
        data sets (re-estimating the joint covariate distribution each
        time), \rev{and} for each conditional model we run knockoffs 10
        times and compute the realized FDP each time. Then averaging
        the 10 FDPs gives an estimate of knockoffs' FDR for that
        conditional model. In an attempt to make each smaller data set
        have size more comparable to the $n\approx 5,000$ in our
        actual experiment, we combined the healthy and CD genetic
        information with that from 5 other diseases from the same data
        set:\footnote{Bipolar disorder was also part of this data set,
          but a formatting error in the data we were given prevented
          us from including it.} coronary artery disease,
        hypertension, rheumatoid arthritis, type 1 diabetes, and type
        2 diabetes. This made for $14,708$ samples altogether. To
        further increase the effective sample size and match the
        actual experiment, we used a random subset of 1000
        observations for choosing cluster representatives, but the
        \emph{same} 1000 for each smaller data set, so that each
        subsampled data set contained $\approx 1,400$ unique samples,
        $+ 1000$ common samples (for each of these common samples
        $\tilde X_{ij} = X_{ij}$). For computational reasons, we used
        only the first chromosome in this experiment, so each of our
        simulations had (pre-clustering) $29,258$ covariates. The
        conditional model for the response was chosen to be a logistic
        regression model with 60 nonzero coefficients of random signs
        and locations uniformly chosen from among the original (not
        just representatives) SNPs. \rev{In the second experiment (Figure~\ref{fig:wtccc_fake}), t}he coefficient amplitude was
        varied, and for each amplitude value, 10 different conditional
        models (different random locations of the nonzero
        coefficients) were simulated. Each simulation ran the exact
        same covariance estimation, SNP clustering and representative
        selection, knockoff construction, and knockoff selection
        procedure as used for the real data. \rev{In the third
          experiment (Figure~\ref{fig:wtccc_fake_tuning}), the
          shrinkage was amplified or attenuated by varying the value
          $m$ in \citet[Equation (2.7)]{XW-MS:2010} from its intended
          value, while the coefficient amplitude was fixed at 14.} In both of these experiments, knockoffs appears to be quite robust to the approximation we make in estimating the joint distribution of the SNPs.
	
	\subsection{Results on real data}
	Encouraged by these robustness simulations, we
	proceeded to run knockoffs with a nominal FDR level of 10\% on the
	full genetic design matrix and real Crohn's disease outcomes. Since
	knockoffs is a randomized procedure, we re-ran knockoffs 10 times
	(after choosing the representatives)
	and recorded the selected SNPs over all repetitions, summarized in
	Table~\ref{tab:CDresults}. The serial computation time for a single
	run of knockoffs was about 6 hours, but the knockoff generation
	process is trivially parallelizable over chromosomes, so with 20
	available computation nodes, the total parallelized computation time was
	about one hour.
	\setlength\extrarowheight{1pt}
	\begin{table}\centering
		\begin{tabular}{|p{1.7cm}|p{2.85cm}|c|c|p{2.2cm}|p{1.8cm}|}
			\hline
			Selection frequency & Cluster \;\;\;Representative (Cluster Size) & Chrom. &
			Position Range (Mb) & Confirmed in \cite{AF-ea:2010}? & Selected in \cite{WTCCC:2007}? \\
			\hline
			100\% & rs11805303 (16) & 1 & 67.31--67.46 & Yes & Yes \\
			\hline
			100\% & rs11209026 (2) & 1 & 67.31--67.42 & Yes & Yes \\
			\hline
			100\% & rs6431654 (20) & 2 & 233.94--234.11 & Yes & Yes \\
			\hline
			100\% & rs6601764 (1) & 10 & 3.85--3.85 & No & No \\
			\hline
			100\% & rs7095491 (18) & 10 & 101.26--101.32 & Yes & Yes \\
			\hline
			90\% & rs6688532 (33) & 1 & 169.4--169.65 & Yes & No \\
			\hline
			90\% & rs17234657 (1) & 5 & 40.44--40.44 & Yes & Yes \\
			\hline
			90\% & rs3135503 (16) & 16 & 49.28--49.36 & Yes & Yes \\
			\hline
			80\% & rs9783122 (234) & 10 & 106.43--107.61 & No & No \\
			\hline
			80\% & rs11627513 (7) & 14 & 96.61--96.63 & No & No \\
			\hline
			60\% & rs4437159 (4) & 3 & 84.8--84.81 & No & No \\
			\hline
			60\% & rs7768538 (1145) & 6 & 25.19--32.91 & Yes & No \\
			\hline
			60\% & rs6500315 (4) & 16 & 49.03--49.07 & Yes & Yes \\
			\hline
			60\% & rs2738758 (5) & 20 & 61.71--61.82 & Yes & No \\
			\hline
			50\% & rs7726744 (46) & 5 & 40.35--40.71 & Yes & Yes \\
			\hline
			50\% & rs4246045 (46) & 5 & 150.07--150.41 & Yes & Yes \\
			\hline
			50\% & rs2390248 (13) & 7 & 19.8--19.89 & No & No \\
			\hline
			50\% & rs7186163 (6) & 16 & 49.2--49.25 & Yes & Yes \\
			\hline
			40\% & rs10916631 (14) & 1 & 220.87--221.08 & No & No \\
			\hline
			40\% & rs4692386 (1) & 4 & 25.81--25.81 & No & No \\
			\hline
			40\% & rs7655059 (5) & 4 & 89.5--89.53 & No & No \\
			\hline
			40\% & rs7759649 (2) & 6 & 21.57--21.58 & Yes* & No \\
			\hline
			40\% & rs1345022 (44) & 9 & 21.67--21.92 & No & No \\
			\hline
			30\% & rs6825958 (3) & 4 & 55.73--55.77 & No & No \\
			\hline
			30\% & rs9469615 (2) & 6 & 33.91--33.92 & Yes* & No \\
			\hline
			30\% & rs4263839 (23) & 9 & 114.58--114.78 & Yes & No \\
			\hline
			30\% & rs2836753 (5) & 21 & 39.21--39.23 & No & No \\
			\hline
			10\% & rs459160 (2) & 1 & 44.75--44.75 & No & No \\
			\hline
			10\% & rs6743984 (23) & 2 & 230.91--231.05 & Yes & No \\
			\hline
			10\% & rs2279980 (20) & 5 & 57.95--58.07 & No & No \\
			\hline
			10\% & rs4959830 (11) & 6 & 3.36--3.41 & Yes & No \\
			\hline
			10\% & rs13230911 (9) & 7 & 1.9--2.06 & No & No \\
			\hline
			10\% & rs7807268 (5) & 7 & 147.65--147.7 & No & No \\
			\hline
			10\% & rs2147240 (1) & 9 & 71.83--71.83 & No & No \\
			\hline
			10\% & rs10761659 (53) & 10 & 64.06--64.41 & Yes & Yes \\
			\hline
			10\% & rs4984405 (3) & 15 & 93.06--93.08 & No & No \\
			\hline
			10\% & rs17694108 (1) & 19 & 38.42--38.42 & Yes & No \\
			\hline
			10\% & rs3932489 (30) & 20 & 15.01--15.09 & No & No \\
			\hline
		\end{tabular}
				\caption{\label{tab:CDresults} SNP clusters discovered to be important for Crohn's disease over 10
			repetitions of knockoffs.
			Clusters not found in \citet{AF-ea:2010} represent promising sites
			for further investigation, especially rs6601764 and rs4692386, whose
			nearest genes have been independendently linked to CD. See
			text for detailed description. SNP positions are as listed
			in the original data, which uses Human Genome Build 35.}
	\end{table}
	Although in this case we certainly do not
	know the ground truth, we can get some sort of confirmation by
	comparing to the results of studies with newer and much larger data
	sets than ours. In particular, we compared with the results of
	\citet{AF-ea:2010}, which used roughly 22,000 cases and 29,000
	controls, or about 10 times the sample size of the WTCCC data. We also
	compare to the \citet{WTCCC:2007} results, where the p-value cutoff
	used was justified as controlling the Bayesian FDR at close to 10\%---the same level we use. We
	consider discovered clusters in different studies to correspond
	(``Yes'' in Table~\ref{tab:CDresults}) if
	their position ranges overlap, and to nearly correspond (``Yes*'' in
	Table~\ref{tab:CDresults}) if the distance from our discovered cluster to the
	nearest cluster in the other study was less than the width of that
	cluster in the other study.
	
	One thing to notice in the table is that a small number of the discovered
	clusters actually overlap with other clusters, specifically the
	clusters represented by rs11805303 and rs11209026 on chromosome 1, and
	rs17234657 and rs7726744 on chromosome 5. And although they don't
	overlap one another, the three nearby clusters represented by rs3135503,
	rs6500315, and rs7186163 on chromosome 16 all overlap the same
	discovered region in \citet{AF-ea:2010}. Although puzzling at first,
	this phenomenon is readily explained by one of four possibilities:
	\begin{itemize}
		\item By construction, the representatives of overlapping clusters are
		not very highly-correlated (less than 0.5), so the fact that
		knockoffs chose multiple clusters in the same region may mean there
		are multiple important SNPs in this region, with one (or more) in
		each cluster. Focusing on the clusters on chromosome 1, the same
		region on the IL23R gene was reported in \citet{AF-ea:2010} to have
		by far the strongest signal (estimated odds ratio of 2.66, next
		highest was 1.53) among all the regions they identified. If we
		conclude from this that the region or gene is fundamentally
		important for Crohn's disease, it stands to reason that mutations at
		multiple nearby but distinct loci could have important detrimental
		effects of their own.
		\item There could be an important SNP located between the two
		clusters, but which was not in the data. Then the two clusters on either
		side would both be conditionally dependent on the
		response, and knockoffs would be correct to reject them.
		\item One or more of the overlapping clusters could be mundane false
		discoveries caused by null covariates that happened to take values
		more conditionally related to the response than would be
		typical. This would be a facet of the data itself, and thus an
		unavoidable mistake.
		\item The tandem discoveries could also be due to a breakdown in our
		covariate distribution estimate. If the covariance between the two
		representatives were substantially underestimated and one had a
		large effect while the other was null, then the Lasso would have a
		(relatively) hard time distinguishing the two original variables,
		but a much easier time separating them from the knockoff of the null
		representative, since it is less correlated with the signal
		variable. As a result, the null variable and its knockoff would not
		be exchangeable as required by knockoffs, and a consistent error
		could be made. However, given that the empirical and estimated
		correlations between the two representatives on chromosome 1 were
		-0.1813 and -0.1811, respectively, and for the two on chromosome 5
		were -0.2287 and -0.2286, respectively, it seems unlikely that we
		have made a gross estimation error.  Also, note that the separation
		of nearly all the discovered regions, along with the simulations of
		the previous subsection, suggest this effect is at worst very small
		among our discoveries.
	\end{itemize}
	Overlapping clusters aside, the knockoffs results display a number of
	advantages over the original marginal analysis in \citet{WTCCC:2007}:
	\begin{itemize}
		\item First, the power is much higher, with \citet{WTCCC:2007}
		making 9 discoveries, while knockoffs made 18 discoveries on
		average, doubling the power.
		\item Quite a few of the discoveries made by knockoffs that were confirmed by
		\citet{AF-ea:2010} were not discovered in \citet{WTCCC:2007}'s
		original analysis.
		\item Knockoffs made a number of discoveries not found in either
		\citet{WTCCC:2007} or \citet{AF-ea:2010}. Of course we expect some
		(roughly 10\%) of these to be false discoveries, particularly towards
		the bottom of the table. However, especially given the evidence from
		the simulations of the previous subsection suggesting the FDR is
		controlled, it is likely that many of these correspond to true
		discoveries. Indeed, evidence from independent studies about
		adjacent genes shows some of the top hits to be promising
		candidates. For example, the closest gene to rs6601764 is KLF6,
		which has been found to be associated with multiple forms of IBD,
		including CD and ulcerative colitis \citep{WG-ea:2016}; and the closest
		gene to rs4692386 is RBP-J, which has been linked to CD through its
		role in macrophage polarization \citep{MB-ea:2013}.
	\end{itemize}
	Note that these benefits required relatively little customization of
	the knockoff procedure. For instance, \citet{WTCCC:2007} used marginal tests specifically
	tailored to SNP case-control data, while we simply used the LCD
	statistic. We conjecture that the careful use of knockoffs by domain
	experts would compound the advantages of knockoffs, as such
	users could devise more powerful statistics and better model/cluster the
	covariates for their particular application.
	
	\section{Discussion}
	\label{sec:discussion}
	
	\newcommand{\Xko}{\tilde{X}}

	This paper introduced a novel approach to variable selection
        in general non-parametric models, which teases apart important
        from irrelevant variables while guaranteeing Type I error
        control. This approach is a significant rethinking of the
        knockoff filter from \citet{RB-EC:2015}. A distinctive feature
        of our approach is that selection is achieved without ever
        constructing p-values.  This is attractive since (1) p-values
        are not needed and (2) it is unclear how they could be
        efficiently constructed, in general. (The conditional
        randomization approach we proposed is one way of getting such
        p-values but it comes at a computational cost.)
	
	\paragraph{Deployment in highly correlated settings} We posed a simple
	question: which variables does a response of interest depend upon?  In
	many problems, there may not be enough ``resolution'' in the data to
	tell whether $Y$ depends on $X_1$ or, instead, upon $X_2$ when the two
	are strongly correlated.  This issue is apparent in our genetic
	analysis of Crohn's disease from Section \ref{sec:realdata}, where
	co-located SNPs may be extremely correlated. In such examples,
	controlling the FDR may not be a fruitful question. A more meaningful
	question is whether the response appears to depend upon a group of
	correlated variables while controlling for the effects of a number of
	other variables (e.g.,~from SNPs located in a certain region of the
	genome while controlling for the effects of SNPs elsewhere on the
	chromosomes). In such problems, we envision applying our techniques to
	grouped variables: one possibility is to develop a model-\revv{X} group
	knockoff approach following \citet{dai2016knockoff}. Another is to
	construct group representatives and proceed as we have done in Section
	\ref{sec:realdata}. It is likely that there are several other ways to
	formulate a meaningful problem and solution.
	
	\paragraph{Open questions}
	Admittedly, this paper may pose more problems than it solves;
	we close our discussion with a few of them below.
	\begin{itemize}
		\item {\em How do we construct {M\revv{X}} knockoffs?} Even though we presented
		a general strategy for constructing knockoffs, we have essentially
		skirted this issue except for the important case of Gaussian
		covariates. It would be important to address this problem, and write
		down concrete algorithms for some specific distributions of features
		of practical relevance.
		
		\item {\em Which {M\revv{X}} knockoffs?} Even in the case of Gaussian
		covariates, the question remains of how to choose
		$\operatorname{corr}(X_j, \tilde{X}_j)$ or, equivalently, the
		parameter $s_j$ from Section \ref{sec:methodology} since
		$\operatorname{corr}(X_j, \tilde{X}_j) = 1 - s_j$. Should we make
		the marginal correlations small? Should we make the partial
		correlations small? Should we take an information-theoretic approach
		and minimize a functional of the joint distribution such as the
		mutual information between $X$ and $\tilde X$?
		
		\item {\em What would we do with multiple M\revv{X} knockoffs?} As suggested
		in \citet{RB-EC:2015}, we could in principle construct multiple
		knockoff variables $(\Xko^{(1)}, \ldots, \Xko^{(d)})$ in such a way
		that the $(d+1)p$-dimensional family
		$(X,\Xko^{(1)}, \ldots, \Xko^{(d)})$ obeys the following extended
		exchangeability property: for any variable $X_j$, any permutation in
		the list $(X_j, \Xko_j^{(1)}, \ldots, \Xko_j^{(d)}$) leaves the
		joint distribution invariant. On the one hand, such constructions
		would yield more accurate information since we could compute, among
		multiple knockoffs, the rank with which an original variable enters
		into a model. On the other hand, this would constrain the
		construction of knockoffs a bit more, perhaps making them less
		distinguishable from the original features. What is the right trade
		off?
		
		Another point of view is to construct several knockoff matrices
		exactly as described in the paper. Each knockoff matrix would yield
		a selection, with each selection providing FDR control as described
		in this paper. Now an important question is this: is it possible to
		combine/aggregate all these selections leading to an increase in
		power while still controlling the FDR?
		
		\item {\em Can we prove some form of robustness?}  Although our
		theoretical guarantees rely on the knowledge of the joint covariate
		distribution, \revv{Section~\ref{sec:robust} showed preliminary examples with remarkable} robustness when this
		distribution is simply estimated from data.  For instance, the
		estimation of the precision matrix for certain Gaussian designs
		seems to have rather secondary effects on FDR and power levels.  It
		would be interesting to provide some theoretical insights into this
		phenomenon.
		
		\item {\em Which feature importance statistics should we use?} The
		knockoff framework can be seen as an inference machine: the
		statistician provides the test statistic $W_j$ and the machine
		performs inference. It is of interest to understand which statistics
		yield high power, as well as design new ones.
		
		\item {\em Can we speed up the conditional randomization testing
			procedure?} Conditional randomization provides a powerful
		alternative method for controlling the false discovery rate in
		model-\revv{X} variable selection, but at a computational cost that is
		currently prohibitive for large problems. However, there exist a
		number of promising directions for speeding it up, including:
		importance sampling to estimate small p-values with fewer
		randomizations, faster feature statistics $T_j$ with comparable or
		higher power than the absolute value of lasso-estimated
		coefficients, and efficient computation re-use and warm starts to
		take advantage of the fact that each randomization changes
		only a single column of the design matrix.
	\end{itemize}
	
	In conclusion, much remains to be done. On the upside, though, we have
	shown how to select features in high-dimensional nonlinear models
	(e.g.,~GLMs) in a reliable way. This arguably is a fundamental
	problem, and it is really not clear how else it could be achieved.

	\bibliography{references}
	\bibliographystyle{apalike}

\appendix

\section{Relationship with the knockoffs of Barber and Cand\`es}
\label{app:detko}
We can immediately see a key difference with the earlier
framework of \citet{RB-EC:2015}. In their work, the design matrix is
viewed as being fixed and setting $\bS = \bX^\top \bX$, knockoff
variables are constructed as to obey
\begin{equation}
	\label{eq:detko}
	\XX^\top \, \XX = \bG, \quad \bG =  \begin{bmatrix} \bS & \bS - \diag\{s\}\\
		\bS - \diag\{s\} & \bS \end{bmatrix}.
\end{equation}
Imagine the columns of $\bX$ are centered, i.e.,~have vanishing means,
so that we can think of $\bX^\top \bX$ as a sample covariance
matrix. Then the construction of Barber and Cand\`es is asking that
the {\em sample} covariance matrix of the joint set of variables obeys
the exchangeability property---i.e.,~swapping rows and columns leaves
the covariance invariant---whereas in this paper, it is the {\em
  population} covariance that must be invariant. In particular, the
{M\revv{X}} knockoffs will be far from obeying the relationship
$\bX^\top \bX = \tilde{\bX}^\top \tilde{\bX}$ required in
\eqref{eq:detko}. To drive this point home, assume
$X \sim \mathcal{N}(0, \bs I)$. Then we can choose
$\tilde{X} \sim \mathcal{N}(0, \bs I)$ independently from $X$, in
which case $\bX^\top \bX/n$ and $\tilde{\bX}^\top \tilde{\bX}/n$ are
two independent Wishart variables. An important consequence of this is
that in the {M\revv{X}} approach, the sample correlation between the $j$th
columns of $\bX$ and $\tilde{\bX}$ will typically be far smaller than
that in the original framework of Barber and Cand\`es. {For example,
  take $n = 3000$ and $p = 1000$ and assume the equicorrelated
  construction from \citet{RB-EC:2015}. Then the sample correlation
  between any variable $\bX_j$ and its knockoff $\tilde{\bX}_j$ will
  be about 0.65 while in the random case, the average magnitude of the
  correlation is about 0.015.}  This explains the power gain of our
new method that is observed in Section~\ref{sec:simulations}.

	\section{Sequential conditional independent pairs algorithm}
	\label{app:scip}
	
	We will prove by induction on $j$ that Algorithm~\ref{alg:sequential}
	produces knockoffs that satisfy the exchangeability property
	\eqref{eq:randomko}.   We prove the result for the discrete case; the
	general case follows the same argument with a slightly more careful
	measure-theoretic treatment using Radon--Nikodym derivatives instead
	of probability mass functions. Below we denote the probability mass
	function (PMF) of $(X_{1:p},\Xko_{1:j-1})$ by
	$\mathcal{L}(X_{\text{-}j},X_j,\Xko_{1:j-1})$.
	
	\begin{indhyp*}
		After $j$ steps, every pair $X_k,\Xko_k$ is exchangeable in the
		joint distribution of $(X_{1:p},\Xko_{1:j})$ for $k= 1, \ldots, j$.
	\end{indhyp*}
	
	By construction, the induction hypothesis is true after 1 step since
	$X_1$ and $\tilde X_1$ are conditionally independent and have the same
	marginal distribution (this implies conditional exchangeability).
	Assuming the induction hypothesis holds until $j-1$, we prove that it
	holds after $j$ steps. Note that by by assumption, $\mathcal{L}$ is
	symmetric in $X_k,\Xko_k$ (that is, the function value remains
	unchanged when the argument values $X_k,\Xko_k$ are swapped) for
	$k = 1,\dots,j-1$.  Also, the conditional PMF of $\Xko_j$ given
	$X_{1:p},\Xko_{1:j-1}$ is given by
	\[
	\frac{\mathcal{L}(X_{\text{-}j},\Xko_j,\Xko_{1:j-1})}{\sum_u\mathcal{L}(X_{\text{-}j},u,\Xko_{1:j-1})}.\]
	Therefore, the joint PMF of $(X_{1:p},\Xko_{1:j})$ is given by the
	product of the aforementioned conditional PMF with the joint PMF of $(X_{1:p},\Xko_{1:j-1})$:
	\begin{equation}
		\label{eq:sciplik}
		\frac{\mathcal{L}(X_{\text{-}j},X_j,\Xko_{1:j-1})\mathcal{L}(X_{\text{-}j},\Xko_j,\Xko_{1:j-1})}{\sum_u\mathcal{L}(X_{\text{-}j},u,\Xko_{1:j-1})}.
	\end{equation}
	Exchangeability of $X_j,\Xko_j$ follows from the symmetry of
	\eqref{eq:sciplik} to those two values. For $k<j$, note that
	\eqref{eq:sciplik} only depends on $X_k,\Xko_k$ through the function
	$\mathcal{L}$, and that $\mathcal{L}$ is symmetric in $X_k,\Xko_k$. Therefore,
	\eqref{eq:sciplik} is also symmetric in $X_k,\Xko_k$, and therefore
	the pair is exchangeable in the joint distribution of
	$(X_{1:p},\Xko_{1:j})$.

	\section{Bayesian knockoff statistics}
	\label{app:bvs}
	
	The data for the simulation of Section~\ref{sec:bvs} was drawn from:
	\begin{align*}
		X_j &\iidsim \mathcal{N}(0,1/n), \; j\in\{1,\dots,p\},\\
		\beta_j &\iidsim \mathcal{N}(0,\tau^2), \; j\in\{1,\dots,p\},\\
		\delta_j &\iidsim \text{Bernoulli}(\pi), \; j\in\{1,\dots,p\},\\
		\frac{1}{\sigma^2} &\sim \text{Gamma}(A,B) \qquad \text{(shape/scale
			parameterization, as opposed to shape/rate)},\\
		Y &\sim \mathcal{N}\left(\sum_{j:\delta_j=1}X_j\beta_j,\;\sigma^2\right).
	\end{align*}
	The simulation used $n = 300$, $p = 1000$, with parameter values $\pi
	= \frac{60}{1000}$, $A = 5$, $B = 4$, and $\tau$ varied along the x-axis of the plot.
	
	To compute the Bayesian variable selection (BVS) knockoff statistic,
	we used a Gibbs sampler on the following model (treating $X_1,\dots,X_p$ and $\tilde{X}_1,\dots,\tilde{X}_p$ as fixed):
	\begin{align*}
		\beta_j &\iidsim \mathcal{N}(0,\tau^2), \; j\in\{1,\dots,p\},\\
		\tilde{\beta}_j &\iidsim \mathcal{N}(0,\tau^2), \; j\in\{1,\dots,p\},\\
		\lambda_j &\iidsim \text{Bernoulli}(\pi), \; j\in\{1,\dots,p\},\\
		(\delta_j,\tilde{\delta}_j) &\iidsim \left\{\begin{array}{ll}
			(0,0) & \text{if } \lambda_j = 0 \\
			(0,1) \text{ w.p. } 1/2 & \text{if } \lambda_j = 1 \\
			(1,0) \text{ w.p. } 1/2 & \text{if }\lambda_j = 1
		\end{array}\right\}
		,\; j\in\{1,\dots,p\},\\
		\frac{1}{\sigma^2} &\sim \text{Gamma}(A,B) \qquad \text{(shape/scale
			parameterization, as opposed to shape/rate)},\\
		Y &\sim \mathcal{N}\left(\sum_{j:\delta_j=1}X_j\beta_j+\sum_{j:\tilde{\delta}_j=1}\tilde{X}_j\tilde{\beta}_j,\;\sigma^2\right),
	\end{align*}
	which requires only a very slight modification of the procedure in
	\citet{EG-RM:1997}. After computing the posterior probabilities
	$\hat{\delta}_j$ and $\hat{\tilde{\delta}}_j$ with 500 Gibbs samples (after 50 burn-in samples), we
	computed the $j$th knockoff statistic as
	\[W_j = \hat{\delta}_j-\hat{\tilde{\delta}}_j.\]
	
	\section{Conditional randomization speedups}
	\label{app:CRspeed}
	In order to make the conditional randomization computation feasible
	for Figure~\ref{fig:rand_v9}, we had to apply a number of
	speed-ups/shortcuts. With these, the serial computation time for the figure was
	reduced to roughly three years.
	\begin{itemize}
		\item As mentioned, the LCD statistic is a powerful one for {M\revv{X}}
		knockoffs, so we wanted to compare it to the analogue for the
		conditional randomization test, namely, the absolute value of the
		Lasso-estimated coefficient. However, choosing the penalty parameter
		by cross-validation turned out to be too expensive, so instead we
		chose a fixed value by simulating repeatedly from the known model
		(with amplitude 30), running cross-validation for each
		repetition, and choosing the median of the
		cross-validation-error-minimizing penalty parameters (the chosen
		value was 0.00053). For a fair comparison, we did the same for
		{M\revv{X}} knockoffs (except the simulations included knockoff
		variables too---the chosen value was 0.00077). This speed-up is of
		course impossible in practice, as the true model is not known. It
		is not clear how one would choose the penalty parameter if not by
		cross-validation (at considerable extra computational
		expense), although a topic of current research is how to choose tuning
		parameters efficiently and without explicit reliance on a model for
		the response, e.g., \citep{JL-JL:2014}.
		\item Because the statistic used was the (absolute value of the)
		estimated coefficient from a sparse estimation procedure, many of
		the observed statistics were exactly zero, and for these, the
		conditional randomization p-values can be set to one without further
		computation. This did not require any prior knowledge, although it
		will only work for statistics whose distribution has a point mass at
		the smallest possible value.
		\item {Because the power was calibrated, we knew to expect at least
			around 10 discoveries, and thus could anticipate the BHq cutoff
			being at least $0.1 \times 10/600$. This cutoff gives a sense of
			the p-value resolution needed, and we chose the number of
			randomizations to be roughly 10 times the inverse of the BHq
			cutoff, namely, $10,000$. However, we made sure that all $10,000$
			randomizations were only computed for very few covariates, both
			using the previous bullet point and also by checking after
			periodic numbers of randomizations whether we can reject the
			hypothesis that the true p-value is below the approximate BHq
			cutoff upper-bound of $0.1 \times 44/600$ (the 44 comes from 40
			nonzeros with 10\% false discoveries). For instance, after just 10
			randomizations, if the approximate p-value so far is greater than
			0.2, we can reject the hypothesis that the exact p-value is below
			$0.1 \times 44/600$ at significance 0.0001 (and thus compute no
			further randomizations for that covariate).}  Speed-ups like this
		are possible in practice, although they require having an idea of
		how many rejections will be made, which is not generally available
		ahead of time.
	\end{itemize}
	
\section{Robustness Simulations}
\label{app:robust}
In Figure~\ref{fig:kn_v59}, the points labeled ``$100\alpha\%$
Emp. Cov'' represents knockoffs run using the following convex
combination of the true and empirical covariance matrices:
\[\bS_{\text{EC}} = (1-\alpha) \bS +
  \alpha \hat\bS,\] where $\hat\bS$ is the empirical covariance
matrix. The point labeled ``Graph. Lasso'' represents knockoffs run
using the following covariance estimate:
\[\bS_{\text{GL}} = \diag(r)\,\hat{\bs{\Theta}}^{-1}\diag(r)\]
where the rescaling vector $r$ has
$r_j = \sqrt{\Sigma_{jj}/(\hat{\bs{\Theta}}^{-1})_{jj}}$, and
$\hat{\bs{\Theta}}$ is the inverse covariance estimated by the
graphical Lasso with penalty parameter chosen by 2-fold
cross-validation on the log-likelihood.

	\section{Genetic analysis of Crohn's disease}
	\label{app:realdata}
	The SNP arrays came from an Affymetrix 500K chip, with calls made by
	the BRLMM algorithm \citet{BRLMM}. SNPs satisfying any
	of the following conditions were removed:
	\begin{itemize}
		\item minor allele frequency (MAF) $< 1\%$,
		\item Hardy--Weinberg equilibrium test p-value  $< 0.01\%$,
		\item missing $>5\%$ of values (values were considered missing if
		the BRLMM score was $> 0.5$),
		\item position was listed as same as another SNP (this occured just
		for the pair rs16969329 and rs4886982; the former had smaller MAF
		and was removed),
		\item position was not in genetic map.
	\end{itemize}
	Furthermore, subjects missing $>5\%$ of values were removed. Following
	the original paper describing/using this data \citep{WTCCC:2007}, we
	did not adjust for population structure. Any missing values that
	remained after the above preprocessing were replaced by the mean of
	the nonmissing values for their respective SNPs.
		
	\section{Knockoffs with selection}
	\label{app:knselect}
	
	First, recall that the results of Theorem~\ref{thm:main} hold if for
	any subset $S \subset \mathcal{H}_0$, we have
	\begin{equation}
		\label{eq:1}
		([\bX, \, \tilde{\bX}]_{\swap(S)}, y) \, \eqd \, ([\bX, \,
		\tilde{\bX}], y).
	\end{equation}
	In fact, M\revv{X} knockoffs are defined in such a way that this property
	holds. Now, the procedure employed in Section~\ref{sec:realdata} to
	construct knockoffs is slightly different from that described in the
	rest of the paper.  Explicitly, the data looks like
	\[
	\bX = \begin{bmatrix} \bX^{(1)} \\ \bX^{(2)} \end{bmatrix}, \quad y
	= \begin{bmatrix} y^{(1)} \\ y^{(2)} \end{bmatrix},
	\]
	where $y^{(1)}$ is $983 \times 1$, $\bX^{(1)}$ is $983 \times 71,145$,
	$y^{(2)}$ is $3930 \times 1$ and $\bX^{(2)}$ is $3930 \times 71,145$;
	recall that the set of samples labeled (1) is used to select the
	cluster representatives, and that the two sets (1) and (2) are
	independent of each other.  The knockoffs $\tilde{\bX}^{(2)}$ for
	$\bX^{(2)}$ are generated as described in the paper (using the ASDP
	construction) and for (1), we set $\tilde{\bX}^{(1)} = \bX^{(1)}$. To
	verify \eqref{eq:1}, it suffices to show that for any subset
	$S\in\mathcal{H}_0$ and each $i \in \{1,2\}$,
	\begin{equation}
		\label{eq:selexch1}
		([\bX^{(i)},\, \tilde{\bX}^{(i)}], y^{(i)}) \, \eqd \, ([\bX^{(i)},\,
		\tilde{\bX}^{(i)}]_{\text{swap}(S)}, y^{(i)})
	\end{equation}
	since the sets (1) and (2) are independent. \eqref{eq:selexch1} holds
	for $i = 2$ because we are following the classical knockoffs
	construction. For $i = 1$, \eqref{eq:selexch1} actually holds with
	exact equality, and not just equality in distribution.
	
	We can also argue from the perspective of M\revv{X} knockoffs from Definition
	\ref{def:def_MfK}. By construction, it is clear that
	$\tilde{\bX}^{(2)}$ are valid M\revv{X} knockoffs for $\bX^{(2)}$. We thus
	study $\tilde{\bX}^{(1)}$: since $\tilde{\bX}^{(1)} = {\bX}^{(1)}$,
	the exchangeability property is trivial; also,
	$\tilde{\bX}^{(1)} \independent y^{(1)} \, | \, \bX^{(1)}$ since
	$\bX^{(1)}$ determines $\tilde{\bX}^{(1)}$.
\end{document}